\documentclass[journal,twocolumn]{IEEEtran}
\usepackage{cite}
\usepackage[cmex10]{amsmath}
\usepackage{amssymb, amsthm}
\usepackage{subeqn}
\usepackage{graphicx}
\usepackage{epstopdf}
\usepackage{xcolor}
\usepackage[font=footnotesize]{subfig}

\newtheorem{prob-statement}{Problem}
\newtheorem{lemma}{Lemma}
\newtheorem{lemma-app}{Lemma}[section]

\newtheorem{thrm}{Theorem}
\newtheorem{prop}{Proposition}

\newtheorem{claim}{Claim}
\DeclareMathOperator*{\argmax}{arg\,max}

\DeclareMathOperator*{\conv}{conv}

\begin{document}

\title{Design of Binary Quantizers for Distributed Detection under Secrecy Constraints}

\author{V.~Sriram~Siddhardh~Nadendla,~\IEEEmembership{Student Member,~IEEE,}
        and~Pramod~K.~Varshney,~\IEEEmembership{Fellow,~IEEE}

\thanks{V. Sriram Siddhardh Nadendla and Pramod K. Varshney are with the Department
of Electrical Engineering and Computer Science, Syracuse University, Syracuse, NY 13244, USA. E-mail: \{vnadendl, varshney\}@syr.edu.}


}


\maketitle


\begin{abstract}
In this paper, we investigate the design of distributed detection networks in the presence of an eavesdropper (Eve). We consider the problem of designing binary quantizers at the sensors that maximize the Kullback-Leibler (KL) Divergence at the fusion center (FC), subject to a tolerable constraint on the KL Divergence at Eve. In the case of i.i.d. received symbols at both the FC and Eve, we prove that the structure of the optimal binary quantizers is a likelihood ratio test (LRT). We also present an algorithm to find the threshold of the optimal LRT, and illustrate it for the case of Additive White Gaussian Noise (AWGN) observation models at the sensors. In the case of non-i.i.d. received symbols at both FC and Eve, we propose a dynamic-programming based algorithm to find efficient quantizers at the sensors. Numerical results are presented to illustrate the performance of the proposed network design. 
\end{abstract}


\begin{IEEEkeywords}
Distributed Detection, Wireless Sensor Networks, Eavesdroppers, Kullback-Leibler Divergence, Secrecy.
\end{IEEEkeywords}


\IEEEpeerreviewmaketitle


\section{Introduction \label{sec: Introduction}}
Distributed detection has been a well-studied topic over the past three decades, with a wide range of applications ranging from civilian to military purposes\cite{Book-Swami,Viswanathan1997,Blum1997,Book-Varshney,Veeravalli2011}. A distributed detection network comprises of a network of spatially distributed sensors that observe the phenomenon-of-interest (PoI) and send processed information to a fusion center (FC) where a global decision is made regarding the presence or absence of the PoI. In order to design a distributed detection network, the designer needs to choose an appropriate set of sensor quantizers and the fusion rule in the network. Tsitsiklis and Athans showed that the joint design of an optimal distributed detection network is NP-Hard \cite{Tsitsiklis1985}, in general. Therefore, the problem is often decomposed into two design problems \cite{Hoballah1986}, where the problems of the design of sensor quantizers and the fusion rule are considered separately. For example, the optimal fusion rule for a set of known and conditionally independent sensor quantizers is given by the Chair-Varshney rule \cite{Chair1986}. In the presence of a large number of sensors where the fusion rule can be abstracted out by adopting error-exponents as performance metrics,
several attempts have been made to analyze and design sensor quantizers in the past \cite{Tsitsiklis1986,Longo1990,Tsitsiklis1993,Hashlamoun1996,Rago1996,Zhang2002,Liu2006,Wang2013} under different scenarios in the absence of an eavesdropper. In this paper, we address the problem of designing optimal local quantizers in the presence of an eavesdropper, with Kullback-Leibler (KL) Divergence as the design metric.


In the past, a few attempts have been made to address the problem of eavesdropping threats by designing ciphers in the broader context of sensor networks. For example, Aysal \emph{et al.} in \cite{Aysal2008} investigated the problem of secure distributed estimation by incorporating a stochastic cipher in the existing sensor networks to improve secrecy. They showed a significant deterioration in Eve's performance (in terms of bias and mean squared error) at the cost of a marginal increase in the estimation variance at the FC. A similar attempt has been made in the context of distributed detection in sensor networks by Nadendla in \cite{Nadendla-Thesis}, where the author presented an optimal network (sensor quantizers, flipping probabilities in the stochastic cipher and the fusion rule) that minimizes the error probability at the FC in the presence of a constraint on Eve's error probability. In \cite{Jeon2011}, Jeon \emph{et al.} proposed a cooperative transmission scheme for a sensor network where the sensors are partitioned into non-flipping, flipping and dormant sets, based on the thresholds dictated by the FC. The non-flipping set of sensors quantize the sensed data and transmit them to the FC, while the flipping sensors transmit flipped decisions in order to confuse the Eve. The sensors within the dormant set sleep, in order to conserve energy and we have an energy-efficient sensor network with longer lifetime.

In all of the above attempts, security in distributed detection systems was incorporated as an afterthought in that separate security blocks were added after the original system had been designed without considering the possible security threats. Marano \emph{et al.} in \cite{Marano2009a}, on the other hand, investigated the problem of designing optimal decision rules for a censoring sensor network in the presence of eavesdroppers. Although their framework of censoring sensor networks is more general, they assume that the Eve can only determine whether an individual sensor transmits its decision or not. In reality, Eve can extract more information than just merely determining the presence or absence of transmission, and hence can make a reasonably good decision regarding the PoI, based on its receptions. Therefore, in our preliminary work in \cite{Nadendla2010a}, we investigated the problem of designing sensor quantizers for a distributed detection network that maximize the difference in the KLDs at the FC and Eve. Note that the objective considered in \cite{Nadendla2010a}, namely the difference in KLDs at the FC and Eve, does not constrain the Eve's performance. Consequently, Eve may acquire an intolerable amount of information from the sensors, and therefore, the solution (quantizer design) provided in \cite{Nadendla2010a} may not be attractive to the network designer in many practical scenarios.

In this paper, we consider a distributed detection network in the presence of binary symmetric channels (BSCs) between the sensors and the FC, as well as those between the sensors and the Eve, whose transition probabilities are known to the network designer. In contrast to our work in \cite{Nadendla2010a} where the goal was to design binary quantizers that maximize the difference in the KLDs at the FC and Eve, in this paper, we design optimal binary sensor quantizers that maximize KL Divergence at the FC while constraining the Eve's KL Divergence to a prescribed tolerance level. We consider two scenarios, one where the channels between the sensors and the FC (likewise, channels between sensors and the Eve) are identical, and the second where the channels are non-identical. In the identical channel scenario, we assume that the Eve has \emph{noisier} channels than the FC's channels, and show that the structure of the optimal quantizer at the local sensors is a \emph{likelihood ratio test} (LRT). We present an illustrative example where we assume that the sensors make noisy observations of a known deterministic signal. We present an algorithm to find the optimal threshold so as to maximize the KL Divergence at the FC while ensuring that the Eve's KL Divergence remains within tolerable limits. In the scenario where channels are non-identical, we decompose the problem into $N$ subproblems to be solved sequentially using dynamic programming. Consequently, we decouple the Eve's constraint into $N$ individual constraints, thus allowing us to solve each of these decoupled problems as in the identical sensor case.


The remainder of the paper is organized as follows. In Section \ref{sec: System model}, we present the system framework, introduce the design metrics and state the problem considered in this paper. Then, in Section \ref{sec: Identical Quantizer Design}, we consider the scenario where all the channels to the FC are identical and so are the channels to Eve. We present fundamental tools regarding transformations in the receiver operating characteristics (ROC) of a given sensor in Appendix \ref{sec: Transformations}. These are necessary to address this scenario. For the sake of illustration, we present an example where we assume that the sensors make noisy observations of a deterministic signal, and present an algorithm to find the optimal threshold for the LRT in the presence of Eve. Numerical results are also presented where we discuss the tradeoff between the network performance and tolerable secrecy. 
In Section \ref{sec: Non-Identical Quantizer Design}, we consider a more general problem setup where the design of non-identical sensor thresholds is considered in the presence of independent, but non-identically distributed sensor observations and non-identical channels. Our concluding remarks are presented in Section \ref{sec: Conclusion}.


\section{System Model and Problem Formulation \label{sec: System model}}

\begin{figure}[!t]
	\centering
    \includegraphics[width=3in]{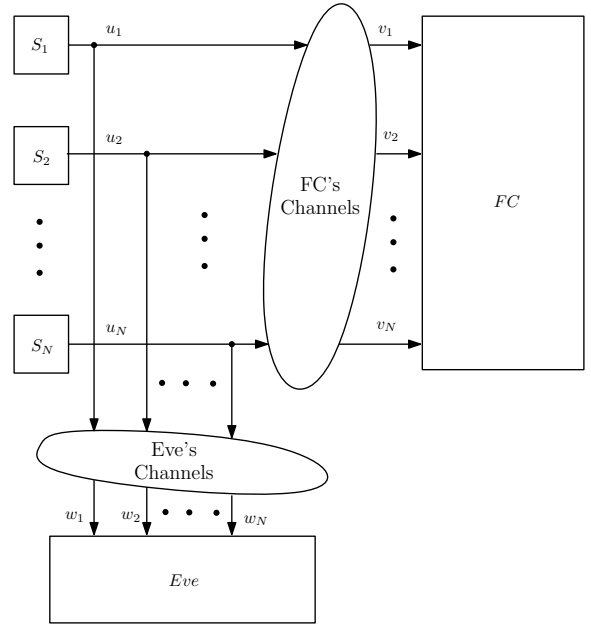}
    \caption{Sensor Network Model}
    \label{Fig: model}
\end{figure}

Consider a binary-hypothesis testing problem for distributed detection with $N$ sensors under the Neyman Pearson framework, as shown in Figure \ref{Fig: model}. Let $\mathbf{r}_i = \{ r_{i,t} : t = 1, \cdots, T \}$ denote a sequence of i.i.d. observations (in time) acquired by the $i^{th}$ sensor over $T$ time periods. Furthermore, we also assume that these observations $\mathbf{r}_i$ are independent across sensors, i.e., for $i = 1, \cdots, N$, but do not necessarily have identical distributions at different sensors. Let $H_0$ and $H_1$ denote the null and the alternate hypotheses respectively. We denote the conditional probability density functions of $r_{i,t}$ under hypotheses $H_0$ and $H_1$ as $p_{i,0}(r) = p(r_{i,t} = r|H_0)$ and $p_{i,1}(r) = p(r_{i,t} = r|H_1)$ respectively. In this paper, for all $i = 1, \cdots, N$, we assume that the $i^{th}$ sensor employs binary quantization to compress its observation $r_{i,t}$ into $u_{i,t}$, as defined below, using a decision rule $\gamma_i(\cdot)$.
\begin{equation}
    u_{i,t} = \gamma_i(r_{i,t}) = 
    \begin{cases}
    	1, & \mbox{where } \Lambda(r_{i,t}) \geq \lambda_i
        \\
       	0, & \mbox{otherwise.}
    \end{cases}
    \label{Eqn: Quantizer-rule}
\end{equation}
where $\Lambda(r_{i,t})$ is a test-statistic and $\lambda_i$ is a suitable threshold to be designed.

Let $x_i = P(u_{i,t} = 1 | H_0)$ and $y_i = P(u_{i,t} = 1 | H_1)$ denote the false-alarm and detection probabilities at the $i^{th}$ sensor respectively. The pair $(x_i, y_i)$ is traditionally referred to, as the operating point of the $i^{th}$ sensor, which can lie anywhere on the compact\footnote{In this context, compactness of the unit-square corresponds to the inclusion of the boundary points (0,0), (0,1), (1,0) and (1,1) within the set itself.} unit-square $\mathcal{U} = [0,1]^2$, which we call the \emph{ROC space}. For any fixed test-statistic $\Lambda(\cdot)$, when the threshold $\lambda_i$ is varied, the operating point of the $i^{th}$ sensor follows a curve $y_i = g_{\Lambda}(x_i)$. This curve $y_i = g_{\Lambda}(x_i)$ is traditionally known as the \emph{ROC curve}. In the rest of the paper, we use the operating point $(x_i,y_i)$ to represent the quantizer rule $\gamma_i$ employed at the $i^{th}$ sensor. Two quantizers $\gamma_1$ and $\gamma_2$ are considered identical (equivalent), if their operating points $(x_1, y_1)$ and $(x_2, y_2)$ are the same. 

Let $\Gamma_i$ denote the set of all feasible\footnote{The feasibility of an operating point is primarily dictated by the quality of the sensing observations. Note that the size of $\Gamma_i$ diminishes as the sensor observations get corrupted due to multipath fading and/or thermal noise.} operating points $(x_i,y_i)$ at the $i^{th}$ sensor. Then, the region $\Gamma_i$ in the ROC space is upper-bounded by the set of operating points corresponding to the likelihood ratio tests (LRTs). We call this boundary as the \emph{LRT curve}, and denote it as $y_i = g_{LRT_i}(x_i)$. Furthermore, we restrict our analysis only to those operating points that lie above the line $y_i = x_i$ in the ROC plane. This is because any point below the line $y_i = x_i$ contributes negatively to the overall performance in terms of error probability at the FC. In summary, the region $\Gamma_i$ in the ROC space is upper-bounded by the LRT curve $y_i = g_{LRT_i}(x_i)$, and lower-bounded by the line $y_i = x_i$.


Given the operating point $(x_i, y_i)$, the Kullback-Leibler (KL) Divergence of the $i^{th}$ sensor is defined as follows.
\begin{equation}
	D_i = \displaystyle x_i \log \frac{x_i}{y_i} + (1 - x_i) \log \frac{1 - x_i}{1 - y_i}
	\label{Eqn: KLD - sensor}
\end{equation}

Let $\Upsilon = \{0,1\}^N$ denote the $N$-dimensional space of compressed symbols $\mathbf{u}_t = \{ u_{1,t}, \cdots, u_{N,t} \}$ at all the sensors at a given time $t$. In this paper, we assume that the $i^{th}$ sensor transmits its compressed symbols $u_{i,t}$ to the FC through a binary-symmetric channel (BSC) with transition probability $\rho_{fc_i}$. In our model, we also assume that an eavesdropper wiretaps each of these sensor transmissions through a BSC with transition probability $\rho_{e_i}$. 

If $\mathbf{v}_i = \{ v_{1, t}, \cdots, v_{N,t} \}$ and $\mathbf{w}_i = \{ w_{1, t}, \cdots, w_{N,t} \}$ denote the received symbols at the FC and Eve respectively, the operating point $(x_i, y_i)$ at the $i^{th}$ sensor gets transformed into $(x_{fc_i}, y_{fc_i})$ and $(x_{e_i}, y_{e_i})$ at the FC and Eve respectively, which are given as follows.
\begin{subequations}
	\begin{equation}
    	\label{Eqn: Eve-ROC-x-fc-def}
    	x_{fc_i} = P(v_{i,t} = 1|H_0) = \rho_{fc_i} + (1 - 2 \rho_{fc_i}) x_i
	\end{equation}
	\begin{equation}
    	\label{Eqn: Eve-ROC-y-fc-def}
    	y_{fc_i} = P(v_{i,t} = 1|H_1) = \rho_{fc_i} + (1 - 2 \rho_{fc_i}) y_i
	\end{equation}
	\begin{equation}
    	\label{Eqn: Eve-ROC-x-e-def}
        x_{e_i} = P(w_{i,t} = 1|H_0) = \rho_{e_i} + (1 - 2 \rho_{e_i}) x_i
	\end{equation}
	\begin{equation}
    	\label{Eqn: Eve-ROC-y-e-def}
        y_{e_i} = P(w_{i,t} = 1|H_1) = \rho_{e_i} + (1 - 2 \rho_{e_i}) y_i
	\end{equation}
\end{subequations}

Let the contributions of the $i^{th}$ sensor to the overall KL Divergence at the FC and Eve be denoted as $D_{FC_i}$ and $D_{E_i}$ respectively. Then, $D_{FC_i}$ and $D_{E_i}$ are defined as follows.
\begin{equation}
    \begin{array}{lcl}
        \displaystyle D_{FC_i} & = & \displaystyle x_{fc_i} \log \left( \frac{x_{fc_i}}{y_{fc_i}} \right) + (1 - x_{fc_i}) \log \left( \frac{1 - x_{fc_i}}{1 - y_{fc_i}} \right)
        \\
        \\
        \displaystyle D_{E_i} & = & \displaystyle x_{e_i} \log \left( \frac{x_{e_i}}{y_{e_i}} \right) + (1 - x_{e_i}) \log \left( \frac{1 - x_{e_i}}{1 - y_{e_i}} \right).
    \end{array}
\end{equation}

Let $\mathcal{A}_T^{FC}, \mathcal{A}_T^E \in \Upsilon^T$ denote the acceptance regions of the hypothesis $H_1$ at FC and Eve respectively, over a time-window $t = 1, \cdots, T$. Then, the global probabilities of false alarm and miss at the FC and Eve are given by
\begin{equation}
    \begin{array}{cc}
        p_T^{FC} = Pr(\mathbf{v_i} \in \mathcal{A}_T^{FC} | H_0), & q_T^{FC} = Pr(\mathbf{v_i} \in \mathcal{\overline{A}}^{FC}_T | H_1).
        \\
        \\
        p_T^E = Pr(\mathbf{w_i} \in \mathcal{A}_T^E | H_0), & q_T^E = Pr(\mathbf{w_i} \in \mathcal{\overline{A}}^E_T | H_1).
    \end{array}
\end{equation}
where $\mathcal{\overline{A}}^{FC}_T$ and $\mathcal{\overline{A}}^{E}_T$ are the rejection regions of the hypothesis $H_1$ at the FC and Eve respectively, and, $\mathbf{v_i} = \{ v_{i,1}, \cdots, v_{i,T} \}$ and $\mathbf{w_i} = \{ w_{i,1}, \cdots, w_{i,T} \}$ are the received symbols at the FC and Eve respectively, transmitted by the $i^{th}$ sensor over a time window of length $T$. Next, we present Stein's Lemma that addresses the asymptotic properties of the global probability of miss $q_{T}^{FC}$.

\begin{lemma}[\emph{Stein's Lemma} \cite{Cover}]
    For any $0 < \delta, \varphi < \frac{1}{2}$, let $q_{T, \delta}^{FC} = \displaystyle \min_{p_T^{FC} < \delta} q_T^{FC}$ and $q_{T, \varphi}^E = \displaystyle \min_{p_T^E < \varphi} q_T^E$. Then, we have
    \begin{equation}
        \begin{array}{lcl}
            \displaystyle \lim_{\delta \rightarrow 0} \lim_{T \rightarrow \infty} - \frac{1}{T} \log q_{T, \delta}^{FC} & = & \mathcal{D}_{FC}
            \\
            \\
            \displaystyle \lim_{\varphi \rightarrow 0} \lim_{T \rightarrow \infty} - \frac{1}{T} \log q_{T, \varphi}^E & = & \mathcal{D}_E
        \end{array}
    \end{equation}
    where $\mathcal{D}_{FC}$ and $\mathcal{D}_E$ are the KL divergences at the FC and Eve respectively, which are defined as follows.
	\begin{equation}
		\begin{array}{lcr}
	    	\displaystyle \mathcal{D}_{FC} = \displaystyle \sum_{i = 1}^N D_{FC_i} & \mbox{and} & \displaystyle \mathcal{D}_E = \displaystyle \sum_{i = 1}^N D_{E_i}.
		\end{array}
		\label{Eqn: KLD-FC-Eve}
	\end{equation}
\end{lemma}

Thus, KL Divergence is the error exponent for the global probability of miss when the global probability of false alarm is constrained (and diminishing to zero with time). Therefore, as a surrogate to the global probability of miss, we choose KL Divergence as the performance metric in this paper. Note that $\mathcal{D}_{FC}$ and $\mathcal{D}_{E}$ are both convex functions of $\mathbf{x} = \{x_1, \cdots, x_N\}$ and $\mathbf{y} = \{y_1, \cdots, y_N\}$ in the hyper-cube $[0,1]^N$, which is made up of the ROC spaces of all the sensors in the detection network.

In this paper, we design a distributed detection network where $\mathcal{D}_{FC}$ is maximized while constraining $\mathcal{D}_E$ to a prescribed tolerance limit, denoted as $\alpha$. 
We present the formal problem statement and discuss the various scenarios that are addressed in this paper, as follows.
\begin{prob-statement}
	Find
    \begin{flalign*}
    	\displaystyle \argmax_{\boldsymbol\gamma} & \quad \mathcal{D}_{FC} \quad \mbox{ s.t.}
        \\ & \text{1. } \quad \mathcal{D}_E \ \leq \ \alpha
        \\ & \text{2. } \quad (x_i, y_i) \in \Gamma_i, \mbox{ for all } i = 1, \cdots, N.
    \end{flalign*}
	\label{Problem: General}
\end{prob-statement}

Note that Constraint 1 in the above problem statement becomes degenerate for large values of $\alpha$. More specifically, Problem \ref{Problem: General} is meaningful only when $0 \leq \alpha < \alpha^*$ so that it has a non-degenerate Constraint 1 in Problem \ref{Problem: General}. This critical value $\alpha^*$ is equal to Eve's KL Divergence $\mathcal{D}_E^*$, which Eve attains when FC attains the maximum KL Divergence $\mathcal{D}_{FC}^*$. This maximum KL Divergence $\mathcal{D}_{FC}^*$ can be found by solving Problem \ref{Problem: General} in the absence of Constraint 1.

Let $\mathcal{R} \triangleq \displaystyle \cap_{i = 1}^N \Gamma_i \cap \{ \ (\mathbf{x},\mathbf{y}) \ | \ \mathcal{D}_E \leq \alpha \ \}$ denote the search space in Problem \ref{Problem: General}. Note that $\{ (\mathbf{x},\mathbf{y}) \ | \ \mathcal{D}_E \leq \alpha \}$ is a convex level-set of $\mathcal{D}_E$ \cite{Book-Rockafeller1996}, because $\mathcal{D}_E$ is a convex function of $(\mathbf{x},\mathbf{y})$. Similarly, since LRTs are optimal in the absence of Eve (For a detailed proof, please refer to Proposition 4.1 in \cite{Tsitsiklis1993}), $\Gamma_i$ is also a convex set in the ROC space. Also, $\mathcal{R}$ is an intersection of two convex sets, and therefore, $\mathcal{R}$ is a convex set. 

Since $\mathcal{D}_{FC}$ is a convex function of $(\mathbf{x}, \mathbf{y})$, Problem \ref{Problem: General} is a convex maximization problem, and therefore, the optimal solution is one of the extreme points of $\mathcal{R}$ \cite{Book-Rockafeller1996}. Note that $\{ (\mathbf{x},\mathbf{y}) \ | \ \mathcal{D}_E \leq \alpha \}$ is not necessarily a subset of $\Gamma_i$, and therefore, the optimal set of binary quantizers need not necessarily be LRTs. 
Furthermore, the search space $\mathcal{R}$ in Problem \ref{Problem: General} is not a simple polytope. $\mathcal{R}$  is an intersection of two convex sets with smooth boundaries and therefore, its boundary does not necessarily have a smooth differential at every point. Consequently, optimal search algorithms proposed to solve traditional convex maximization problems with polytope search spaces cannot be applied to find the optimal solution of Problem \ref{Problem: General}, as our problem demands a more detailed analysis of the boundary of the search space. 


Therefore, in Section \ref{sec: Identical Quantizer Design}, we first restrict our attention to a simpler scenario\footnote{In this paper, we call this scenario as ``identical sensors and channels".} where all the sensors' observations are identically distributed and, where all the channels between the sensors and the FC (likewise, channels between sensors and the Eve) are identical. This assumption results in the received symbols at the FC (likewise, received symbols at the Eve) being conditionally i.i.d., thus decomposing the problem into a distributed framework of $N$ identical sub-problems. In Section \ref{sec: Non-Identical Quantizer Design}, we consider a more general scenario\footnote{Similarly, we call this scenario as ``non-identical sensors and channels".} where the sensor observations are conditionally independent and non-identically distributed, and the channels between the sensors and the FC (likewise, channels between sensors and the Eve) are also non-identical. In both these scenarios, we investigate the design of secure binary quantizers when $\alpha < \mathcal{D}_E^*$.



\section{Optimal Quantizer Design in the Presence of Identical Sensors and Channels \label{sec: Identical Quantizer Design}}
In this section, we address the problem of designing optimal quantizers when all the sensors and the channels between the sensors and the FC (likewise, channels between sensors and the Eve) are identical.

For all $i = 1, \cdots, N$, we have
\begin{equation}
    \begin{array}{ll}
        \displaystyle p_{i,0}(x) = p_0(x), & \displaystyle p_{i,1}(x) = p_1(x)
        \\
        \displaystyle x_i = x, & \displaystyle y_i = y 
        \\
        \displaystyle \rho_{fc_i} = \rho_{fc}, & \displaystyle \rho_{e_i} = \rho_e
    \end{array}
\end{equation}

Since all the sensors and their corresponding channels are identical, we remove the sensor-indices for notational simplicity. Therefore, we have $x_{fc_i} = x_{fc}$, $y_{fc_i} = y_{fc}$, $x_{e_i} = x_e$ and $y_{e_i} = y_e$ for all $i = 1, \cdots, N$. Because of this, $D_i = D$, $D_{FC_i} = D_{FC}$ and $D_{E_i} = D_E$ for all $i = 1, \cdots, N$, and consequently, the KLD at the FC and Eve reduces to $\mathcal{D}_{FC} = N D_{FC}$ and $\mathcal{D}_E = N D_E$. In other words, Problem \ref{Problem: General} reduces to the design of the quantizer at one of the identical sensors as follows.

\begin{prob-statement}
	\label{Problem: Identical}
	Find
    \begin{flalign*}
    	\displaystyle \argmax_{\gamma} & \quad D_{FC} \quad \mbox{ s.t.}
        \\ & \text{1. } \quad D_E \ \leq \ \tilde{\alpha}
        \\ & \text{2. } \quad (x, y) \in \Gamma.
    \end{flalign*}
\end{prob-statement}
where $\tilde{\alpha} = \displaystyle \frac{\alpha}{N}$.

Note that, although Problem \ref{Problem: Identical} is still a convex maximization problem, due to its reduced dimensionality, the problem becomes tractable. In the remaining section, we find the optimal quantizer in two stages. First, we find the structure of the optimal binary quantizers by gaining insights into the behavior of $D_{FC}$ on the boundary of the Eve's constraint $\{(x,y) \ | \ D_E \leq \tilde{\alpha}\}$. Then, we present an algorithm to find the optimal threshold for this quantizer.

We start our investigation of the behavior of $D_{FC}$ on the boundary of the Eve's constraint $\{ (x,y) \ | \ D_E \leq \tilde{\alpha} \}$ by determining the necessary conditions for guaranteeing $D_E = \tilde{\alpha}$ in the following lemma.


\begin{lemma}
	If the transition probability of the Eve's BSCs satisfies $\rho_e < \displaystyle \frac{1}{2}$, the two necessary conditions for any sensor operating point $(x,y)$ to guarantee $D_E = \tilde{\alpha}$ in the ROC space are stated as follows.
	\begin{equation}
		\displaystyle \frac{dy}{dx} = \frac{\log \left( \displaystyle \frac{1-x_e}{1-y_e} \right) - \log \left(\displaystyle \frac{x_e}{y_e} \right)}{\displaystyle \frac{1-x_e}{1-y_e} - \frac{x_e}{y_e}}
		\label{Eqn: diff_y_x}
	\end{equation}
	and
	\begin{equation}
        \begin{array}{l}
            \displaystyle \left( \frac{1-x_e}{1-y_e} - \frac{x_e}{y_e} \right) \frac{d^2y}{dx^2}
            \\
            \\
            \qquad \qquad = \displaystyle (1-2 \rho_e) \left[ - \left( \frac{1-x_e}{(1-y_e)^2} + \frac{x_e}{y_e^2} \right) \left(\frac{dy}{dx}\right)^2 \right.
            \\
            \\
            \qquad \qquad \quad \displaystyle \left. + 2 \left( \frac{1}{y_e} + \frac{1}{1-y_e} \right) \frac{dy}{dx} - \left( \frac{1}{x_e} + \frac{1}{1-x_e} \right) \right].
        \end{array}
		\label{Eqn: diff2_y_x}
	\end{equation}
	\label{Lemma: Eve_constraint}
\end{lemma}

\begin{proof}
	Since $D_E$ is a constant (equal to the fixed design-parameter $\tilde{\alpha}$), its first two derivatives are equal to zero. We employ these to prove the lemma.
	
	First, we differentiate $D_E$ with respect to $x$ and equate it to zero, as follows.
	\begin{equation}
		\begin{array}{lcl}
			\displaystyle \frac{dD_E}{dx} & = & \displaystyle \frac{d}{dx} \left[ x_e \log \frac{x_e}{y_e} + (1-x_e) \log \left( \frac{1-x_e}{1-y_e} \right) \right]
			\\
			\\
			& = & \displaystyle (1-2\rho_e) \left[ \left( \frac{1-x_e}{1-y_e} - \frac{x_e}{y_e} \right) \frac{dy}{dx} \right.
            \\
            \\
            && \qquad \displaystyle \left. - \left\{ \log \left( \frac{1-x_e}{1-y_e} \right) - \log \left( \frac{x_e}{y_e}\right) \right\} \right]
			\\
			\\
			& = & 0.
		\end{array}
		\label{Eqn: diff_D_E_x}
	\end{equation}	
	Rearranging the terms in Equation \eqref{Eqn: diff_D_E_x}, we can obtain Equation (\ref{Eqn: diff_y_x}).
	
	Next, we differentiate Equation (\ref{Eqn: diff_D_E_x}) again with respect to $x$ as follows, in order to find a closed-form expression for $\displaystyle \frac{d^2y}{dx^2}$.
	\begin{equation}
		\begin{array}{lcl}
			\displaystyle \frac{d^2D_E}{dx^2} & = & \displaystyle (1-2\rho_e) \frac{d}{dx} \left[ \left( \frac{1-x_e}{1-y_e} - \frac{x_e}{y_e} \right) \frac{dy}{dx} \right.
            \\
            \\
            && \quad \displaystyle \left. - \left\{ \log \left( \frac{1-x_e}{1-y_e} \right) - \log \left( \frac{x_e}{y_e}\right) \right\} \right]
			\\
			\\
			& = & \displaystyle (1-2\rho_e) \left[ \left( \frac{1-x_e}{1-y_e} - \frac{x_e}{y_e} \right) \frac{d^2y}{dx^2} \right.
            \\
            \\
            && \quad \displaystyle \left. + (1-2 \rho_e) \left( \frac{1-x_e}{(1-y_e)^2} + \frac{x_e}{y_e^2} \right) \left(\frac{dy}{dx}\right)^2 \right.
			\\
			\\
			&& \qquad \left. \displaystyle - 2 (1-2\rho_e) \left( \frac{1}{y_e} + \frac{1}{1-y_e} \right) \frac{dy}{dx} \right.
            \\
            \\
            && \qquad \quad \displaystyle \left. + (1-2\rho_e) \left( \frac{1}{x_e} + \frac{1}{1-x_e} \right) \right].
			\\
			\\
			& = & 0.
		\end{array}
		\label{Eqn: diff2_D_E_x}
	\end{equation}
	Rearranging the terms in Equation \eqref{Eqn: diff2_D_E_x}, we can obtain Equation (\ref{Eqn: diff2_y_x}).
\end{proof}

Note that Equation \eqref{Eqn: diff_D_E_x} in Lemma \ref{Lemma: Eve_constraint} provides the slope of the Eve's constraint boundary $D_E = \tilde{\alpha}$. Since the slope of $y$ with respect to $x$ along the boundary $D_E = \tilde{\alpha}$ has a structure similar to the slope of a line joining two points on a logarithmic curve as seen in Equation \eqref{Eqn: diff_y_x}, we present lower and upper bounds for the slope of this boundary curve $D_E = \tilde{\alpha}$ in the ROC plane in the following lemma.


\begin{lemma}
	\label{Lemma: Eve_constraint_2}
	The slope of the Eve's constraint boundary in the ROC plane, as defined by the set of points $\{ \ (x,y) \ | \ D_E = \tilde{\alpha} \ \}$, is bounded on both sides as follows.
	\begin{equation}
		\displaystyle \frac{x_e}{y_e} \leq \frac{dy}{dx} \leq \frac{1-x_e}{1-y_e}.
		\label{Eqn: diff_y_x_bound}
	\end{equation}
\end{lemma}
\begin{proof}
	Given two points $a \geq b$, due to the concavity of the $\log(\cdot)$ function, the slope of the line joining $(a, \log a)$ and $(b, \log b)$ always lies between the slopes of the $\log (\cdot)$ at points $a$ and $b$ respectively 
	Hence, this results in Equation \eqref{Eqn: diff_y_x_bound}.	
\end{proof}

Note that the necessary conditions for any operating point $(x,y)$ to lie on the Eve's constraint boundary $\{ \ (\mathbf{x},\mathbf{y}) \ | \ D_E = \tilde{\alpha} \ \}$, as stated in Lemma \ref{Lemma: Eve_constraint}, and the bounds on the slope of the same boundary curve, as given in Lemma \ref{Lemma: Eve_constraint_2}, are essential to our analysis of the behavior of the sensor's KL divergence $D$, and the FC's KL Divergence, $D_{FC}$, in terms of the false alarm probability $x$ along the Eve's constraint, which is defined by $D_E = \tilde{\alpha}$. 

First, we investigate the behavior of the KL Divergence at the sensor, which is denoted as $D(x,y)$, along the Eve's constraint $D_E(x,y) = \tilde{\alpha}$. Note that this analysis can be equivalently interpreted as the case where we investigate the behavior of $D_{FC}$ when the channels between the sensors and the FC are ideal. In the following proposition, we prove that $D(x,y)$ is a convex function of $x$ along the curve $D_E(x,y) = \tilde{\alpha}$.

\begin{prop}
	\label{Thrm: conjecture-theorem}
    Given that the Eve's channel is a BSC with transition probability $\rho_e < \frac{1}{2}$, $D$ is strictly a convex function of $x$, for all operating points that lie in the set $\{ (x,y) \ | \ D_E = \tilde{\alpha} \}$.
\end{prop}

\begin{proof}
    Proof is provided in Appendix \ref{Proof: Thrm-conjecture}.
\end{proof}

For any general BSC between the sensors and the FC, the sensor's operating point $(x,y)$ transforms linearly into $(x_{fc}, y_{fc})$. Consequently, we have the following proposition, where we analyze the behavior of $D_{FC}$ for any general BSC.



\begin{prop}
	\label{Thrm: conjecture-theorem-nonideal}
    Let the BSCs corresponding to the FC and Eve have transition probabilities $0< \rho_{fc}, \rho_e < \frac{1}{2}$. Then, $D_{FC}$ is strictly a convex function of $x$, for all operating points that lie in the set $\{ (x,y) \ | \ D_E = \tilde{\alpha} \}$.
\end{prop}
\begin{proof}
    Note that $(x_{fc}, y_{fc})$ is a linear transformation of $(x,y)$. This can be mathematically expressed as follows.
    \begin{equation}
        \left[ \begin{array}{c} x_{fc} \\ y_{fc} \end{array} \right] = \rho_{fc} \left[ \begin{array}{c} 1 \\ 1 \end{array} \right] + (1 - 2 \rho_{fc}) \left[ \begin{array}{c} x \\ y \end{array} \right].
        \label{Eqn: Affine-Transformation-FC}
    \end{equation}
    In other words, a composition of $D$ with an affine transformation, as given in Equation (\ref{Eqn: Affine-Transformation-FC}), results in $D_{FC}$. Consequently, since $D$ is a convex function, $D_{FC}$ is also a convex function \cite{Book-Boyd2004}.
\end{proof}

Thus, for any BSC with transition probability $\rho_{fc}$ corresponding to the FC, $D_{FC}$ is a convex function of $x$. In other words, among the set of operating points that lie on the Eve's constraint boundary $D_E = \tilde{\alpha}$, the quantizers that maximize $D_{FC}$ always lie on the intersection of the LRT curve $y = g_{LRT}(x)$ and the Eve's constraint boundary $D_E = \tilde{\alpha}$. As a consequence, the optimal quantizer is LRT-based, which we state in the following theorem.

\begin{thrm}
	The optimal quantizer that maximizes the FC's KL Divergence $D_{FC}$ in the presence of a constraint on Eve's KL Divergence $D_E = \tilde{\alpha}$ is a likelihood ratio quantizer.
	\label{Thrm: LRT-Optimal}
\end{thrm}
\begin{proof}
	Let $\mathcal{R}_i \triangleq \Gamma_i \cup \{ (x,y) \ | \ D_E = \tilde{\alpha} \}$ denote the search space in Problem \ref{Problem: Identical}. We know, from Proposition \ref{Thrm: conjecture-theorem}, that $D_{FC}$ is convex with respect to $x$ along the Eve's constraint boundary on the ROC plane. Therefore, the solution of Problem \ref{Problem: Identical} always lies on the extreme points of the set of operating points on the Eve's constraint boundary $\{ (x,y) \ | \ D_E = \tilde{\alpha} \}$. Note that the region of the Eve's constraint boundary that lies within $\mathcal{R}_i$ depends on the choice of $\tilde{\alpha}$. 
	
	Let $D_E^*$ be the maximum KL Divergence at the Eve when the sensor employs the optimal solution to the unconstrained problem where Constraint 1 is not considered in Problem \ref{Problem: Identical}. In the regard, the following two cases arise:
	\begin{itemize}
		\item{ Case-1 [ $\tilde{\alpha} \geq D_E^*$ ]: Note that, $\Gamma_i \subseteq \{ (x,y) \ | \ D_E \leq \tilde{\alpha} \}$ in this case because the Eve's KL Divergence is always within the tolerable limit when the sensor employs any operating point $(x,y) \in \Gamma_i$. Therefore, the solution to Problem \ref{Problem: Identical} is the optimal LRT in this case \cite{Tsitsiklis1993}. 
		}
		\item{ Case-2 [  $\tilde{\alpha} \leq D_E^*$ ]: This is equivalent to the case where $\Gamma_i \nsubseteq \{ (x,y) \ | \ D_E \leq \tilde{\alpha} \}$. Note that we also have $\Gamma_i \not\supset \{ (x,y) \ | \ D_E \leq \tilde{\alpha} \}$ since there always exist operating points $(x,y) \in \Gamma_i$ such that $D_E \leq \tilde{\alpha}$. Therefore, the boundaries of $\Gamma_i$ and $\{ (x,y) \ | \ D_E \leq \tilde{\alpha} \}$ both intersect each other. As discussed earlier in this proof, since the optimal solution is an extreme point of the Eve's constraint boundary $D_E = \tilde{\alpha}$, this is one of the intersection points that also lies on the boundary of $\Gamma_i$. In other words, the optimal sensor quantizer that solves Problem \ref{Problem: Identical} is a LRT.
		}
	\end{itemize}
\end{proof}

As discussed in the proof of Theorem \ref{Thrm: LRT-Optimal}, the problem of finding the optimal quantizer reduces to the problem of finding the intersection points of the boundaries of $\Gamma_i$ and the Eve's constraint $\{ (x,y) \ | \ D_E \leq \tilde{\alpha} \}$, and thereby, finding the corresponding threshold for the optimal LRT at the sensor.


\subsection{Algorithm to find the Optimal Threshold}

Let $f(x) \triangleq D_{FC}(x,y = g_{LRT}(x))$. For the sake of tractability, we consider the problem of finding optimal thresholds when $f(x)$ is a quasi-concave\footnote{Note that
\begin{equation}
    \begin{array}{ll}
        \displaystyle \lim_{x \rightarrow 0} f(x) = 0, & \displaystyle \lim_{x \rightarrow 1} f(x) = 0
    \end{array}
\end{equation}
Since, KLD is always non-negative, we always have $f(x) \geq 0$. Also, since any LRT curve $y = g_{LRT}(x)$ cuts through the level-sets of $D_{FC}$ and is concave, $f(x)$ is a quasi-concave function of $x$.} function of $x$. As shown in Proposition \ref{Thrm: conjecture-theorem}, since the Eve's constraint translates into the convexity of $D_{FC}$ with respect to $x$, there are at most two points of intersection for the curves $y = g_{LRT}(x)$ and $D_E = \tilde{\alpha}$, of which, one of them corresponds to the optimal quantizer. We present this formally in the following claim.

\begin{claim}
    Let $f(x) \triangleq D_{FC}(x, y = g_{LRT}(x))$. If $f(x)$ is a quasi-concave function of $x$, then there are at most two intersection points for the curves $y = g_{LRT}(x)$ and $D_E = \tilde{\alpha}$. The optimal quantizer corresponds to one of the two intersection points.
    \label{Claim: Intersection Points}
\end{claim}

Therefore, the problem reduces to finding these two intersection points and comparing them with respect to each other in terms of their respective $D_{FC}$. Moreover, we wish to find the threshold $\lambda^*$ for the LRT that maximizes $D_{FC}$ in the presence of Eve's constraint. Since, both $x$ and $y$ are tail-probabilities where the start of the tail is the threshold, $x$ and $y$ are both monotonically decreasing functions of the threshold $\lambda$. Therefore, we have the following claim.

\begin{claim}
    The two intersection points can be found by investigating the zeros of the function $h(\lambda) \triangleq D_E(x(\lambda), y(\lambda)) - \tilde{\alpha}$, where $x$ and $y$ are parameterized by the LRT threshold $\lambda$.
    \label{Claim: Zeros of h}
\end{claim}

Let $\tilde{\alpha}_{max}$ denote the value of KL Divergence at which $D_E$ reaches its maximum value. In other words, the optimal quantizer in the absence of Eve (equivalent to $\tilde{\alpha} = \infty$), denoted as the operating point $(x_{\infty}, y_{\infty})$, is the same as the optimal quantizer for any $\tilde{\alpha} \geq \tilde{\alpha}_{max}$. Obviously, the function $h(\lambda)$ has two real zeros only when $\tilde{\alpha} < \tilde{\alpha}_{max}$. Note that only one of them provides the maximum KL Divergence at the FC.

In order to find both zeros of the function $h(\lambda) = 0$, we use the bisection method where we first find the point $\lambda^*$ at which $h(\lambda)$ attains its maximum value. Then, consider two points, one on either side of $\lambda^*$ (which are at a significant distance from $\lambda^*$) as initial points and use the bisection algorithm to find the roots of $h(\lambda) = 0$. We call these two zeros as $\lambda_1$ and $\lambda_2$. Then, we compute and compare $D_{FC}$ at the operating points $(x(\lambda_1), y(\lambda_1))$ and $(x(\lambda_2), y(\lambda_2))$. We choose that threshold as the optimal choice, which results in the maximum $D_{FC}$.

For the sake of illustration, we present an example where the sensors observe the presence or absence of a known deterministic signal, which is corrupted by additive Gaussian noise.


\subsection{Illustrative Example \label{sec: Example}}
We have so far shown that the optimal quantizer lies at the intersection of the curves $D_E = \tilde{\alpha}$ and the LRT boundary in the ROC. But, the structure of the LRT is specific to the observation model, and therefore, it is difficult to characterize the optimal sensor quantizer, in general. Therefore, we illustrate the design methodology for an example, where the sensors observe the presence or absence of a known deterministic signal. In other words, the observations at the $i^{th}$ sensor are modeled as follows.

\begin{equation}
    r_{i,t} = \begin{cases} n_{i,t} & \mbox{if } H_0 \\ \theta + n_{i,t} & \mbox{if } H_1 \end{cases}
\end{equation}
where $\theta$ is the signal-of-interest and $n_{i,t} \sim \mathcal{N}(0, \sigma^2)$ is the additive Gaussian noise with zero mean and variance $\sigma^2$. Then, the probabilities of false alarm and detection are given by
\begin{equation}
    x = \displaystyle Q\left( \frac{\lambda}{\sigma} \right), \ y = \displaystyle Q\left( \frac{\lambda - \theta}{\sigma} \right)
    \label{Eqn: Parameterization}
\end{equation}
where $Q(\cdot)$ is the tail probability of the standard normal distribution $\mathcal{N}(0,1)$.

Substituting Equation \eqref{Eqn: Parameterization} in Equation \eqref{Eqn: KLD - sensor}, we obtain the KL Divergence at the sensor, which is observed to be concave for this example.
Therefore, as stated in Claim \ref{Claim: Intersection Points}, the optimum quantizer is given by the intersection of the LRT boundary in the ROC with the Eve's constraint $D_E = \tilde{\alpha}$.

Note that Equation \eqref{Eqn: Parameterization} is a parameterization of the LRT boundary, where both the ROC's coordinates are parameterized with the threshold of the LRT. Since we are interested in the intersection of the LRT's boundary in the ROC with the Eve's constraint $D_E = \tilde{\alpha}$, we substitute $x_e = \displaystyle \rho_e + (1 - 2 \rho_e) Q\left( \frac{\lambda}{\sigma} \right)$ and $y_e = \displaystyle \rho_e + (1 - 2 \rho_e) Q\left( \frac{\lambda - \theta}{\sigma} \right)$ in $D_E$ to obtain $h(\lambda) = D_E(x(\lambda), y(\lambda)) - \tilde{\alpha}$.

As shown in Figure \ref{Fig: h-vs-lambda}, $h(\lambda)$ is a quasi-concave function of $\lambda$, with the tails converging to $-\tilde{\alpha}$. In other words, there are at most two zero-crossings since the function $h(\lambda)$ is unimodal with the two tails converging to a value less than zero. Therefore, there are at most two solutions to the equation $h(\lambda) = 0$. The optimum sensor threshold can be found by investigating the two zeros of $h(\lambda)$, as suggested in Claim \ref{Claim: Zeros of h}, and comparing them in terms of $D_{FC}$.

\begin{figure}[!t]
	\centering
    \includegraphics[width=3.3in]{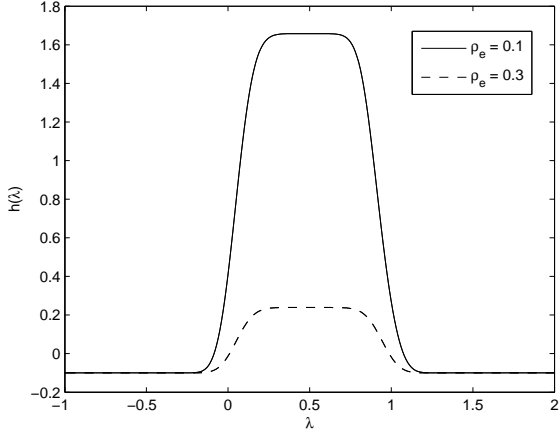}
    \caption{Plot of $h(\lambda)$ as a function of $\lambda$}
    \label{Fig: h-vs-lambda}
\end{figure}



\subsection{Discussion and Results}
In this subsection, we first discuss the impact of the secrecy constraint on the performance of the sensor network. Obviously, when we consider $\tilde{\alpha} = 0$, the network achieves perfect secrecy. But, this also forces the network to be \emph{blind} in that $D_{FC} \rightarrow 0$. On the other extreme, consider a scenario where $\tilde{\alpha} \rightarrow \infty$. This is equivalent to the case where there is no eavesdropper present in the network. In other words, the optimal quantizer is given by $(x_{\infty}, y_{\infty})$. For any finite $\tilde{\alpha} > 0$, we numerically investigate the tradeoff between secrecy and performance of a given distributed detection system.
\begin{figure}[!t]
	\centering
    \includegraphics[width=3.3in]{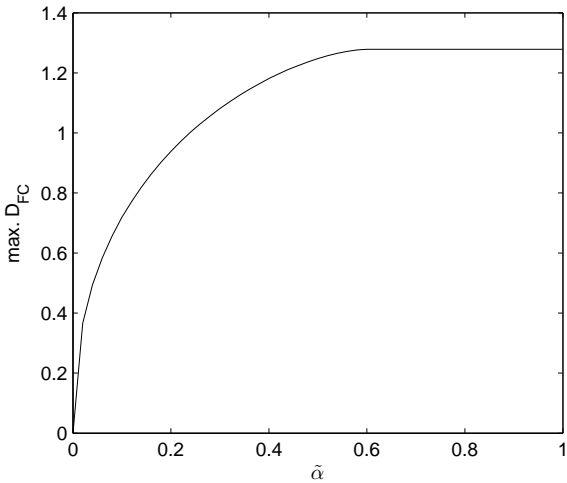}
    \caption{Tradeoff between maximum $D_{FC}$ and $\tilde{\alpha}$.}
    \label{Fig: Tradeoff}
\end{figure}


Since $\tilde{\alpha}$ is the tolerable limit on the performance of Eve, the greater the information leakage we can tolerate, the better the performance of the distributed detection network. This tradeoff is captured by Figure \ref{Fig: Tradeoff}, where the maximum $D_{FC}$ in the presence of a constrained Eve increases with increasing $\tilde{\alpha}$. Note that, beyond a certain value of $\tilde{\alpha}$, the maximum $D_{FC}$ gets saturated to the optimal KLD at the FC in the absence of Eve. This saturation level for this example is 5.8 and it is dictated by the fundamental limits enforced by the imperfect observations and channel models within the network.

Next, we demonstrate the impact of the Eve's constraint on the ROC, as well as the KL Divergence at the FC, in Figure \ref{Fig: Eve-constrained}, when the FC's channels are ideal ($\rho_{fc} = 0$). Note that this argument can be carried over to any general BSC at the FC, as the operating point $(x_{fc}, y_{fc})$ is a linear transformation of $(x,y)$. In Figure \ref{Fig: Eve-constrained}, we assume $\rho_e = 0.1$ and consider two different values of $\tilde{\alpha}$. In Figure \ref{Fig: ROC-constrained}, we plot the constraint curve $D_E = \tilde{\alpha}$ along with the sensor's ROC. Note that the constraint curve intersects the LRT curve at two distinct points, as stated earlier. One of these two intersection points (the intersection point to the right, in this example) is optimal, as shown in Figure \ref{Fig: D_FC-constrained}. Note that the skewness in the ellipses in Figure \ref{Fig: D_FC-constrained} is due to the asymmetry in the KL divergence. Also, as $\tilde{\alpha}$ decreases, $D_{FC}$ becomes deeper and flat-bottomed as a function of $x$ over the Eve's constraint curve $D_E = \tilde{\alpha}$. Another important observation to be made is the fact that the optimal solution in the presence and absence of Eve (red curves) always is on the boundary of the LRT curve, although the thresholds vary depending on the scenario. Since the sufficient test-statistic is the same irrespective of the presence or absence of Eve, the network designer may implement the system in terms of a threshold that can be varied.

\begin{figure*}[!t]
    \centerline
    {
        \subfloat[Sensor's ROC in the presence of Eve]{\includegraphics[width=3.3in]{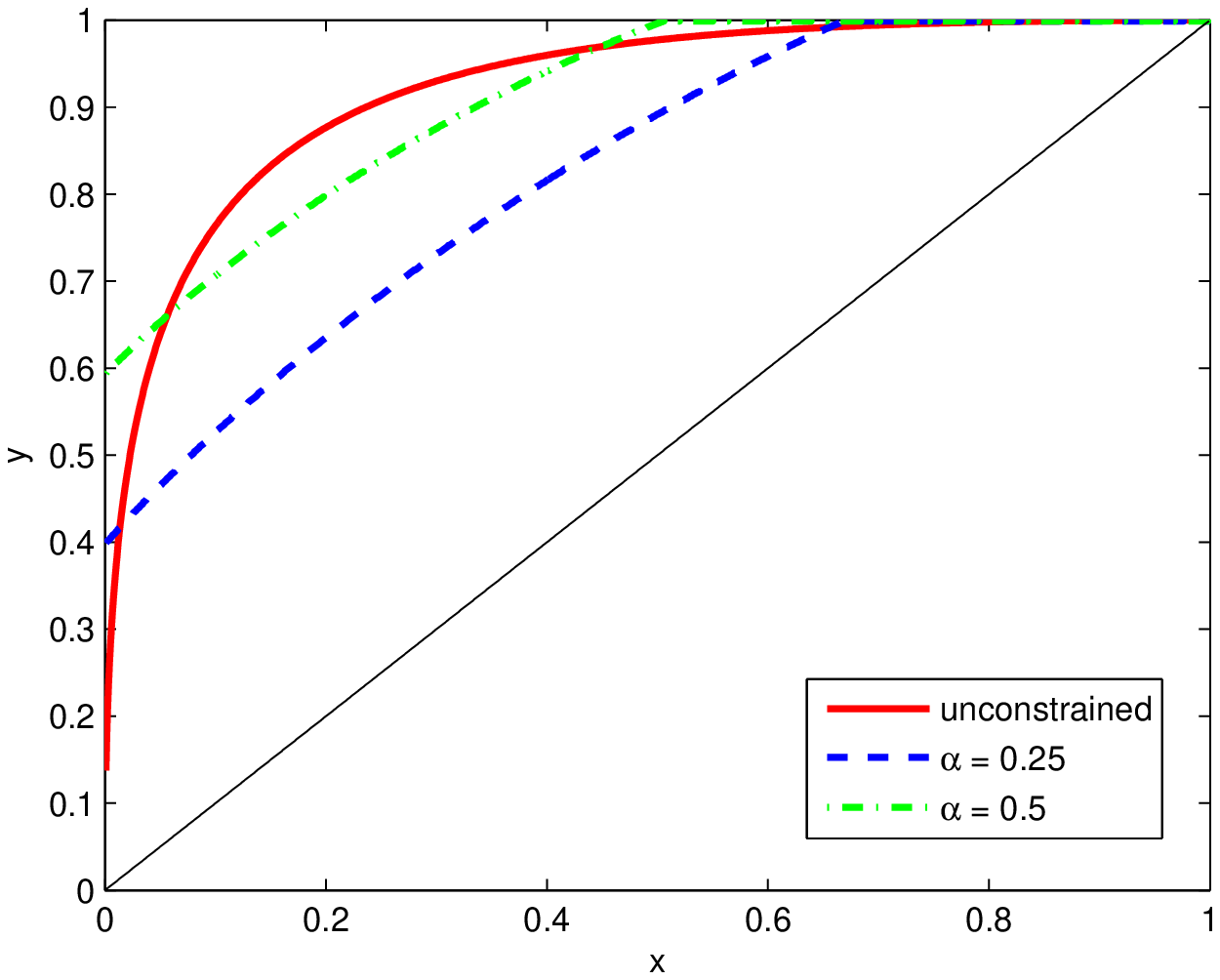}%
        \label{Fig: ROC-constrained}}
        \hfil
        \subfloat[$D_{FC}$ as a function of $x$]{\includegraphics[width=3.3in]{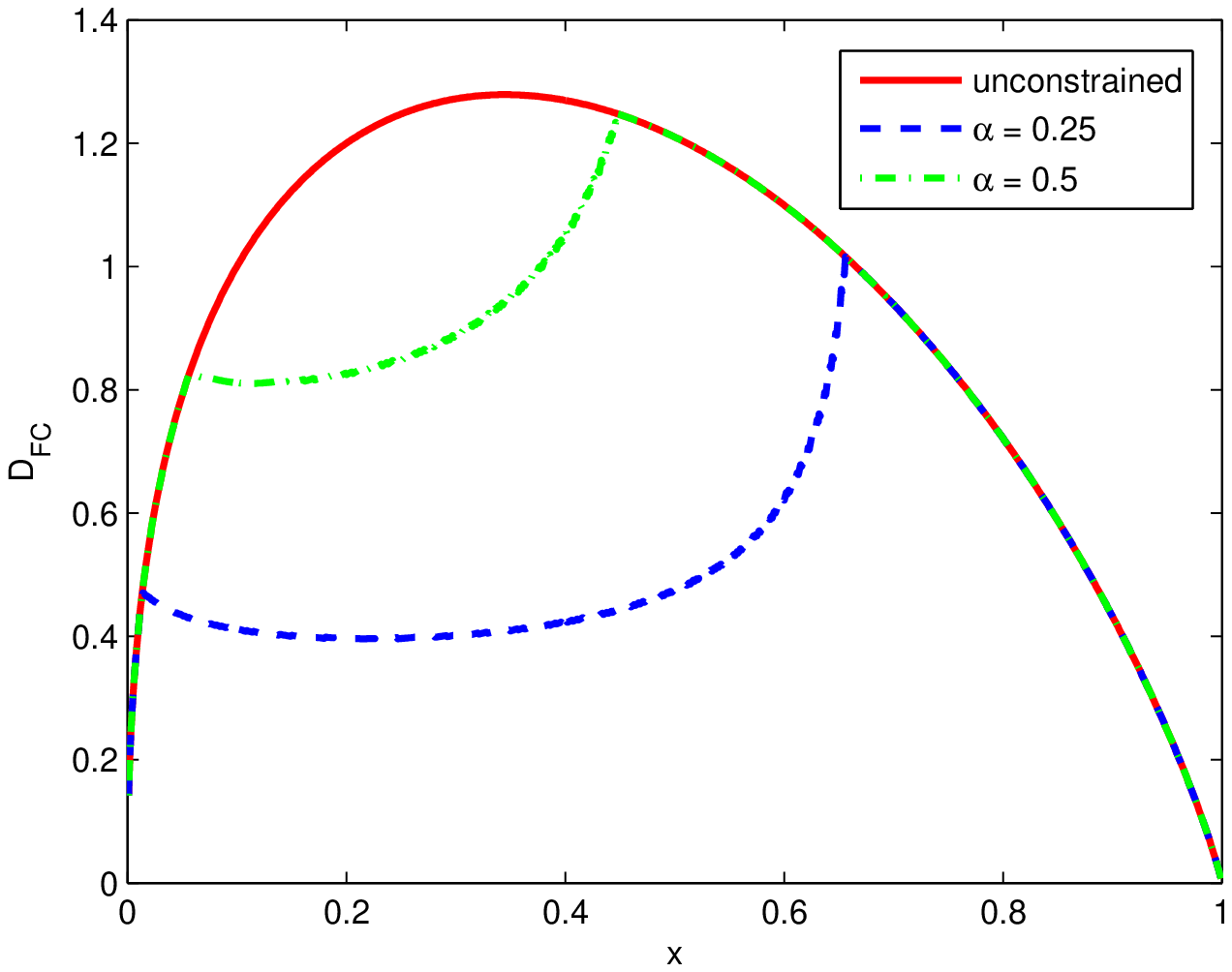}%
        \label{Fig: D_FC-constrained}}
    }
    \caption{Sensor performance in the presence of a constraint, $D_E \leq \tilde{\alpha}$, where $\rho_e = 0.1$.}
    \label{Fig: Eve-constrained}
\end{figure*}



In practice, there exist many conditional probability distributions $p_0(r)$ and $p_1(r)$ for which the computation of likelihood-ratios is intractable. Also, there may be situations where these distributions are not even known to the network designer. In both these cases, the network designer may choose to employ a tractable test that is not LRT. 

Let $\Lambda$ be the test-statistic employed in the sensor quantizer $\gamma$, as defined in Equation \eqref{Eqn: Quantizer-rule}. Note that, by allowing randomization (linear stochastic combination of operating points) between quantizers, Carath\`{e}odary's theorem \cite{Book-Rockafeller1996} and Lemma \ref{Lemma: ROC-Operating-Point-Transformation} in Appendix \ref{sec: Transformations} together makes every operating point $(x,y)$ inside the set $\Psi \triangleq \conv \left(\{ (x,y) \ | \ y \leq g_{\Lambda}(x) \} \right)$ feasible, where $\conv(\mathcal{S})$ represents the convex-hull of a given set $\mathcal{S}$. 


Since $\Psi_i$ is convex, all of our arguments presented in Section \ref{sec: Identical Quantizer Design} also hold for the case of any general non-LRT quantizer. We summarize this in the following claim:
\begin{claim}
	Given any ROC curve $y = g_{\Lambda}(x)$ based on a test-statistic $\Lambda$, the optimal quantizer that maximizes the FC's KL Divergence $D_{FC}$ in the presence of a constraint on Eve's KL Divergence $D_E = \tilde{\alpha}$ within the set $\tilde{\Psi}_i \triangleq \conv \{ (x,y) \ | \ y \leq g_{\Lambda}(x) \}$ always lies on the boundary of $\tilde{\Psi}_i$.
\end{claim}

As discussed earlier in this subsection, this optimal operating point can be implemented by randomizing over a finite set of quantizers, all defined using the same test statistic $\Lambda$.


\section{Efficient Quantizer Design in the Presence of Non-Identical Sensors and Channels \label{sec: Non-Identical Quantizer Design}}

In Section \ref{sec: Identical Quantizer Design}, we investigated the case of identical sensors and channels which was similar to the case of designing the quantizer at a single sensor. In this section, we investigate Problem \ref{Problem: General} when the network has non-identical sensors and/or has non-identical channels. Since Problem \ref{Problem: General} is NP-Hard in general, we propose an efficient methodology for quantizer design that satisfies the Eve's constraint $\mathcal{D}_E \leq \alpha$.

%

Note that the objective function $\mathcal{D}_{FC}$ is linearly separable since the sensor observations are conditionally independent. Therefore, we define
\begin{equation}
    \Phi_n = \Phi_{n-1} + D_{FC_n}, \ \forall \ n = 2, \cdots, N.
    \label{Eqn: Decomposition}
\end{equation}
where $\Phi_1 = D_{FC_1}$. If, at any given intermediate stage, if $\Phi_{n-1}$ is a constant, then the problem of maximizing $\Psi_n$ reduces to the problem of maximizing $D_{FC_n}$.

This above property of KL Divergence at the FC can be used to decompose Problem \ref{Problem: General} into $N$ sub-problems by breaking down the Eve's constraint parameter $\alpha$ into $\boldsymbol\alpha = \{ \alpha_1, \cdots, \alpha_N \}$ using dynamic programming \cite{Book-Bellman} that resembles the \emph{waterfilling} algorithm. Here, for the sake of ensuring the feasibility of our solution, we assume the following.
$$\displaystyle \sum_{i = 1}^N \alpha_i \leq \alpha.$$

Therefore, for a given $\boldsymbol\alpha$, Problem \ref{Problem: General} becomes:
\begin{prob-statement}
	\label{Problem: General-Decomposed}
	For every $i = 1, \cdots, N$, find
    \begin{flalign*}
    	\displaystyle \argmax_{\boldsymbol\gamma} & \quad D_{FC_i} \quad \mbox{ s.t.}
        \\ & \text{1. } \quad D_{E_i} \ \leq \ \alpha_i
        \\ & \text{2. } \quad (x_i, y_i) \in \Gamma_i, \mbox{ for all } i = 1, \cdots, N.
    \end{flalign*}
\end{prob-statement}

Note that the performance of this proposed design-methodology completely depends on the choice of $\boldsymbol\alpha = \{ \alpha_1, \cdots, \alpha_N \}$. To be more precise, the exact solution to Problem \ref{Problem: General} can be equivalently expressed in terms of an optimal decomposition of $\alpha$ into $\boldsymbol\alpha = \{ \alpha_1, \cdots, \alpha_N \}$. Since the problem of finding optimal $\boldsymbol\alpha$ is intractable, we present a suboptimal (greedy) algorithm to find an efficient decomposition of $\alpha$ as follows.

Let $D_{FC_i}^*$ denote the maximum KL Divergence achievable at the FC, due to the $i^{th}$ sensor. In such a setting, Eve attains a KL Divergence $D_{E_i}^*$ due to the $i^{th}$ sensor. We define the quality of the FC's and the Eve's channels corresponding to the $i^{th}$ sensor as $k_i = \displaystyle \frac{D_{FC_i}^*}{D_{E_i}^*}$. The quality $k_i$ represents the tradeoff between the detection performance and secrecy. Let the sensors be ordered in terms of the increasing quality as $k_{i_1} \geq \cdots \geq k_{i_N}$. In other words, we obtain the best tradeoff in terms of the sensor quality by considering sensors in the order of decreasing quality in our sequential allocation mechanism. Therefore, we propose a greedy decomposition of Problem \ref{Problem: General-Decomposed} into $N$ sequential problems based on the sensors' quality, where $\boldsymbol\alpha = \{ \alpha_1, \cdots, \alpha_N \}$ is chosen such that $\mathcal{D}_{FC}$ is maximized in the presence of Eve's constraint $\mathcal{D}_E \leq \alpha$. Note that this decoupling of $\alpha$ into $\boldsymbol\alpha$ allows us to solve each of the individual problems in Problem \ref{Problem: General-Decomposed} using the same method as presented in Section \ref{sec: Identical Quantizer Design}.

Having ordered the nodes in terms of decreasing $k_i^*$, we know that node $i$ achieves better tradeoff than node $j$, if $i > j$. This allows us to select nodes with lower indices to achieve the best tradeoffs between detection performance and secrecy until the resource (constraint on Eve, $\alpha$) is completely utilized. Therefore, the decomposition of $\mathcal{D}_{FC}$, as shown in Equation \eqref{Eqn: Decomposition}, allows us to sequentially select the individual sensors in an increasing order of indices. Therefore, for index $i = 1$, we allocate $\alpha_1 = D_{E_1}^*$ if $\alpha \geq D_{E_1}^*$. Otherwise, $\alpha_1 = \alpha$. Having allocated the Eve's constraint to Sensor 1, we move to Sensor 2. Now, the remaining tolerable leakage information at the Eve is given by $[\alpha - D_{E_1}^*]_+$, where $[x]_+ = x$ if $x \geq 0$, or, $0$ otherwise. Therefore, we solve the problem at Sensor 2 with a new constraint $[\alpha - D_{E_1}^*]_+$.

As the process of selecting the nodes progresses, we reach a point where $N^*$ sensors are already selected and the remaining resource left, given by $\alpha - \displaystyle \sum_{i=1}^{N^*} D_{E_i}^*$, is less than $D_{E_{N^* + 1}}$. Therefore, we let $\alpha_{N^* + 1} = \alpha - \displaystyle \sum_{i=1}^{N^*} D_{E_i}^*$ and let the remaining sensors sleep in order to satisfy the secrecy constraint.

\subsection{Numerical Results}
In order to illustrate the performance of the proposed algorithm, we consider a simple example where, for each $i = 1, \cdots, N$, the $i^{th}$ sensor's observation follows $\mathcal{N}(0, \sigma^2)$ under hypothesis $H_0$ and $\mathcal{N}(\mu_i, \sigma^2)$ under hypothesis $H_1$. Note that this example demonstrates a scenario where the signal source is spaced at different distances from different sensors in the network, and the sensor observations are modelled using a path-loss attenuation channel model. In such a case, the detection probability at the $i^{th}$ sensor can be defined as $y_i = Q\left( Q^{-1}(x) - \eta_i \right)$ in terms of the false alarm probability $x_i$, where $\eta_i = \frac{\mu_i}{\sigma}$ is the corresponding SNR. Assuming that the FC has a perfect channel ($\rho_{fc_i} = 0$), while the Eve has a binary symmetric channel with transition probability $\rho_{e_i} = \rho_i$ at the $i^{th}$ sensor, we have $x_{fc_i} = x_i$, $y_{fc_i} = y_i$, $x_{e_i} = \rho_i + (1 - 2 \rho_i) x_i$ and $y_{e_i} = \rho_i + (1 - 2 \rho_i) y_i$. Then, the KL divergences at the FC and Eve are computed as shown in Equation \eqref{Eqn: KLD-FC-Eve}.


%

\begin{figure*}[!t]
	\centerline
    {
    	\subfloat[KLD vs. N]{\includegraphics[width=3.3in]{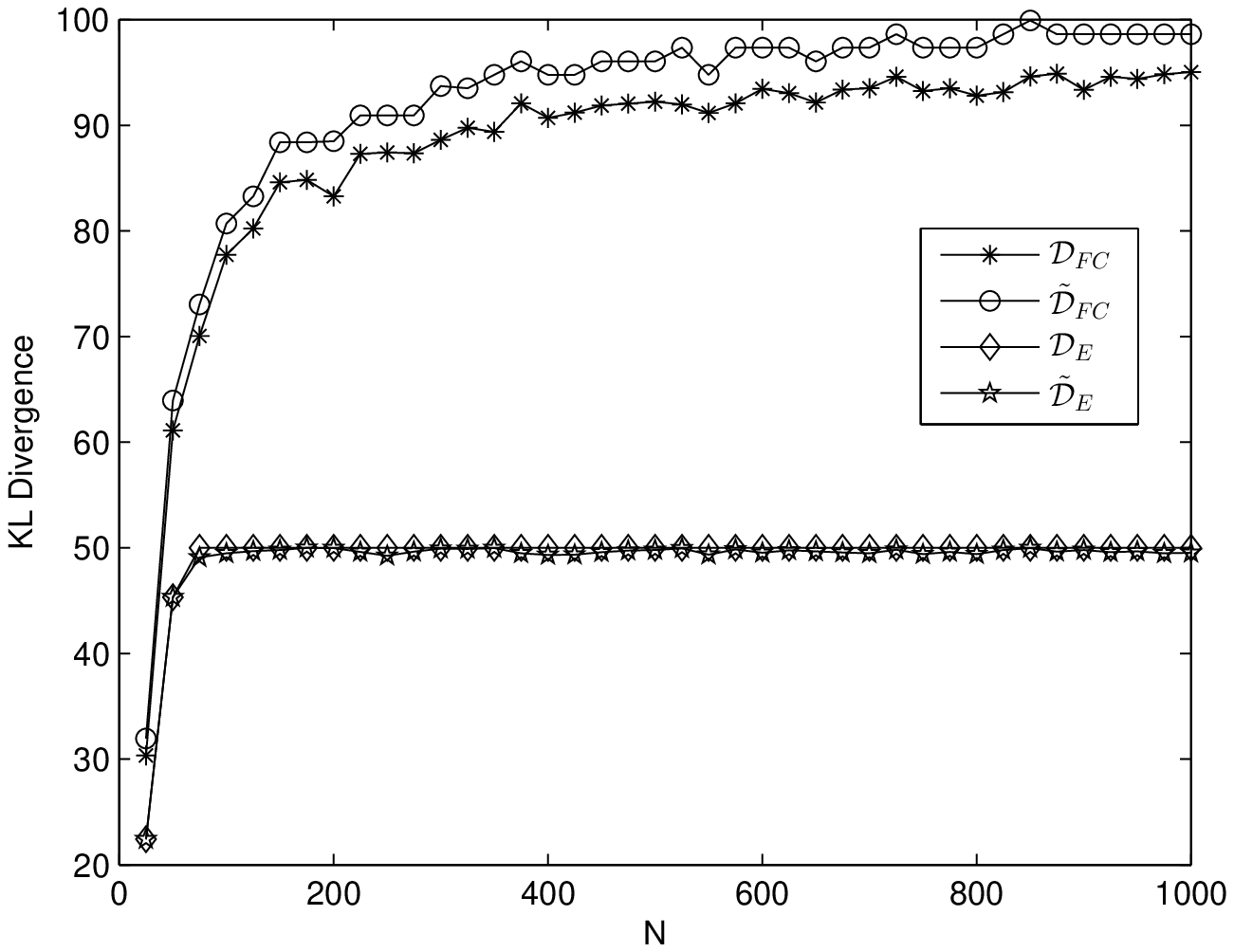}%
        \label{Fig: KLD-Greedy}}
        \hfil
        \subfloat[Number of Active Sensors vs. N]{\includegraphics[width=3.3in]{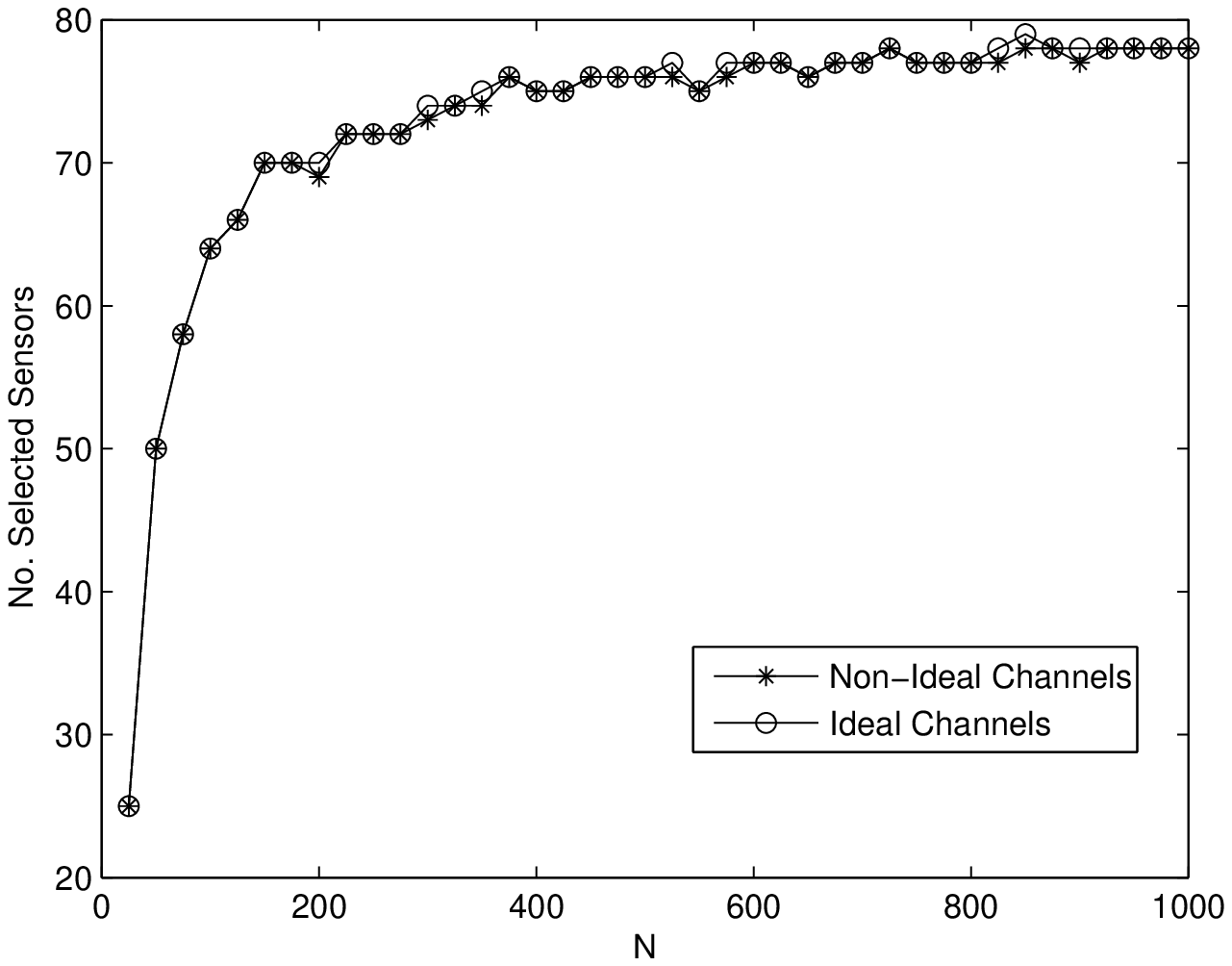}%
        \label{Fig: Active-Sensors}}
    }
    \caption{Performance of the Proposed Greedy Algorithm in a Distributed Inference Network when $\alpha = 50$.}
    \label{Fig: Sensor-Activity}
\end{figure*}

For the sake of illustration, we consider a specific example in order to demonstrate the performance of the proposed greedy algorithm. We assume that all the sensors have identical sensing channels by letting $\eta_i = 1$, for all $i = 1, \cdots, N$. The transition probabilities of the BSCs between the sensors and the FC are sampled randomly from a uniform distribution $\mathcal{U}(0, 0.01)$. Similarly, we let the Eve's channels' transition probabilities be sampled randomly from a uniform distribution $\mathcal{U}(0, 0.1)$. We present a single run of our simulation results in Figure \ref{Fig: Sensor-Activity}, where we present both the KL Divergence at the FC and Eve, along with the number of sensors selected in the network, as a function of $N$ when $\alpha = 50$. Note that, for $\alpha = 50$, the difference between the KL divergences between the FC and Eve is about 40 units. We also provide an upper bound on this difference using a benchmark comparison where we present the case where the FC has ideal channels. In the case where FC has ideal channels, the KL Divergences at the FC and Eve are denoted as $\tilde{\mathcal{D}}_{FC}$ and $\tilde{\mathcal{D}}_E$ respectively. Although the FC's KL divergence is always lower-bounded by Eve's KL divergence, the difference in the KL Divergences at the FC and Eve depend on the quality of the channels at both FC and Eve.

Also, note that, in Figure \ref{Fig: KLD-Greedy}, as the number of sensors increases, both $\mathcal{D}_{FC}$ and $\mathcal{D}_E$ monotonically increase until $N$ reaches a critical point where $\mathcal{D}_E = \alpha$. Beyond this critical point, the algorithm starts to select only those sensors that are prioritized according to the decreasing order of $k_i$. Furthermore, in Figure \ref{Fig: Active-Sensors}, the number of selected sensors increases with increasing number of sensors in the network at the similar rate as that of $\mathcal{D}_{FC}$. Lastly, note that the performance of the distributed inference network in terms of KL Divergence saturates as $N$ increases as per intuition.


\section{Conclusion \label{sec: Conclusion}}
In this paper, we investigated the problem of designing secure binary quantizers in a distributed detection network in the presence of binary symmetric channels. In the case of i.i.d. received symbols at the FC (likewise, i.i.d. received symbols at the Eve), we proved that LRTs are optimal in the presence of a tolerable constraint on Eve's performance. We proposed an algorithm to find the optimal LRT threshold, and presented numerical results to illustrate the performance of our network design. Furthermore, we also proposed efficient quantizer designs in the general case of non-i.i.d. received symbols at the FC and Eve by decomposing the original problem into $N$ sub-problems using a dynamic programming approach. Numerical results were presented to illustrate the efficiency of our proposed algorithm. In our future work, we will investigate the optimal decomposition of the original problem in the general case of non-i.i.d. received symbols at the FC and Eve. In addition, we will also investigate various mitigation schemes that enhance the performance of a distributed detection network in the presence of eavesdroppers.

\appendices

\section{Linear Transformations of Sensor Operating Points in the ROC Space \label{sec: Transformations}}
In this Appendix, we focus our attention on the transformation of the operating point of a single sensor due to the presence of a binary symmetric channel (BSC) between a given sensor and both the FC, as well as between the same sensor and Eve. Let the operating point of a given quantizer be $A = (x,y)$. As mentioned earlier, the sensor's quantizer characteristics $(x,y)$ are represented using its operating point in the sensor's ROC. Also, consider two BSCs with transition probabilities $\rho_1$ and $\rho_2$, each of which transforms the operating point $A = (x,y)$ into $B_1 = (x_1,y_1)$ and $B_2 = (x_2,y_2)$. Let $C = \left( \frac{1}{2}, \frac{1}{2} \right)$. In the following lemma, we present a useful relationship between $A$, $B_1$, $B_2$ and $C$.


\begin{lemma-app}
	Let $0 \leq \rho_1 \leq \rho_2 \leq \frac{1}{2}$. Then, $B_1$ and $B_2$ always lie on the line segment joining $A$ and $C$. In addition, the following inequality holds true.
	\begin{equation}
    	\label{Eqn: xy-inequality}
        \displaystyle \frac{x}{y} \leq \frac{x_1}{y_1} \leq \frac{x_2}{y_2} \leq 1 \leq \frac{1-x_2}{1-y_2} \leq \frac{1-x_1}{1-y_1} \leq \frac{1-x}{1-y}
	\end{equation}
	\label{Lemma: ROC-Operating-Point-Transformation}
\end{lemma-app}

\begin{proof}
    Consider a BSC with transition probability $\rho$, which transforms the operating point $A = (x,y)$ into $B = (\hat{x},\hat{y})$. Then, the equation of the line joining $A$ and $B$ is given by
    \begin{equation}
    	\displaystyle \frac{b-y}{a-x} = \frac{b-\hat{y}}{a-\hat{x}}
    	\label{Eqn: Line}
    \end{equation}
    where $(a,b)$ is some arbitrary point on the line.

    Substituting $\hat{x} = \rho + (1 - 2 \rho) x$ and $\hat{y} = \rho + (1 - 2 \rho) y$, we have
    \begin{equation}
    	\displaystyle \frac{b-y}{a-x} = \frac{b - \rho - (1 - 2 \rho) y}{a - \rho - (1 - 2 \rho) x}.
    	\label{Eqn: Line-2}
    \end{equation}

    Rearranging the terms in Equation \eqref{Eqn: Line-2}, we have
    \begin{equation}
    	\displaystyle (b-y)[a - \rho - (1 - 2 \rho) x] = (a-x)[b - \rho - (1 - 2 \rho) y].
    	\label{Eqn: Line-3}
    \end{equation}

    Simplifying Equation \eqref{Eqn: Line-3}, we have
    \begin{equation}
    	\displaystyle (a-b) + (y - x) = 2(ay - bx).
    	\label{Eqn: Line-4}
    \end{equation}

    Note that the line $a = b$ represents the set of operating points for which the KL Divergence becomes zero. Therefore, let us investigate the point where Equation \eqref{Eqn: Line} intersects the line $a = b$. Substituting $b = a$, we have
    $$\displaystyle (2 a - 1) (y-x) = 0.$$

    In other words, the line in Equation \eqref{Eqn: Line} intersects line $a = b = \frac{1}{2}$ for any transition probability $\rho$. In other words, the points $A$, $B_1$, $B_2$ and $C$ are collinear.

    \begin{figure}[!t]
    	\centering
        \includegraphics[width=3.5in]{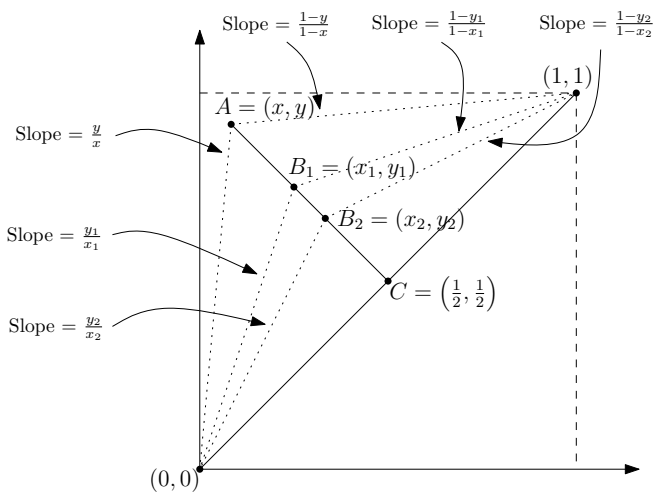}
        \caption{Transformations in the ROC}
        \label{Fig: Transformation}
    \end{figure}

    In fact, as $\rho \rightarrow \frac{1}{2}$, $B \rightarrow C$. In other words, for a given sensor's operating point $A$, the transformed operating point $B$ slides along the line segment joining $A$ and $C$. This sliding behavior can be investigated by analyzing the distance between $B$ and $C$, in terms of increasing $\rho$, as shown in Figure \ref{Fig: Transformation}. We denote the Euclidian distance between $B$ and $C$ as $\phi_{BC} = \sqrt{\left( \hat{x} - \frac{1}{2} \right)^2 + \left( \hat{y} - \frac{1}{2} \right)^2}$. Differentiating $\phi_{BC}$ with respect to $\rho$, we have
    \begin{equation}
        \begin{array}{lcl}
            \displaystyle \frac{d \phi_{BC}}{d \rho} & = & \displaystyle \frac{1}{\phi_{BC}} \left[ \left(\hat{x} - \frac{1}{2}\right)(1-2x) + \left(\hat{y} - \frac{1}{2}\right)(1-2y) \right]
            \\
            \\
            & = & \displaystyle \frac{-1 + \rho + (1 - 2\rho) [x(1-x) + y(1-y)]}{\phi_{BC}}
            \\
            \\
            & = & \displaystyle - \left( \frac{\rho + (1 - 2\rho) [1 - x(1-x) - y(1-y)]}{\phi_{BC}} \right)
            \\
            \\
            & \leq & 0,
        \end{array}
    \end{equation}
    since the function $x(1-x) + y(1-y)$ is concave and attains a maximum value of $\frac{1}{2}$ at $\left( \frac{1}{2}, \frac{1}{2} \right)$. In other words, $B$ slides towards $C$ as $\rho$ increases. Consequently, as shown in Figure \ref{Fig: Transformation}, $B_1$ is farther away from $C$ than $B_2$ on the line joining $A$ and $C$, since $0 \leq \rho_1 \leq \rho_2 \leq 1$.

    Note that the slope of the line joining $(0,0)$ and $B_1$ is $\frac{y_1}{x_1}$, and similarly, $\frac{y_2}{x_2}$ in the case of $B_2$. Since $B_2$ is closer to $B_1$ to $C$, as shown in Figure \ref{Fig: Transformation}, $\frac{y_1}{x_1} \geq \frac{y_2}{x_2}$ and the slope tends to 1 as the transition probability approaches $\frac{1}{2}$. A similar argument holds for the slope of the lines that join $B_1$ and $B_2$ with $(1,1)$. Therefore, the inequality given in Equation (\ref{Eqn: xy-inequality}) holds.
\end{proof}

In order to understand the impact of this transformation on the performance of the network, let us now analyze the KL Divergence at some arbitrary operating point $B = (\hat{x}, \hat{y})$ due to a BSC with transition probability $\rho$ operating on the sensor operating point $A$. In the following lemma, we show that the KL Divergence decreases with increasing $\rho$.

\begin{lemma-app}
    Given the sensor operating point $A = (x,y)$, let $B = (\hat{x}, \hat{y})$ denote the transformed operating point due to a BSC with transition probability $\rho$. Let $D_B$ denote the KL Divergence at $B$. Then, for $0 \leq \rho \leq \frac{1}{2}$, $D_B$ is a monotonically decreasing function of $\rho$ whenever $y \geq x$.
    \label{Lemma: ROC-Operating-Point-Transformation-KLD}
\end{lemma-app}
\begin{proof}
    The KL Divergence at the transformed operating point $B$ is defined as follows.
    \begin{equation}
        D_B =  \displaystyle \hat{x} \log \frac{\hat{x}}{\hat{y}} + (1 - \hat{x}) \log \frac{1 - \hat{x}}{1 - \hat{y}}.
    \end{equation}

    Differentiating $D_B$ with respect to $\rho$, we have
    \begin{equation}
        \begin{array}{lcl}
            \displaystyle \frac{dD_B}{d\rho} & = & \displaystyle (1 - 2y) \left[ \frac{1 - \hat{x}}{1 - \hat{y}} - \frac{\hat{x}}{\hat{y}} \right]
            \\
            && \displaystyle \qquad - (1 - 2x) \left[ \log \left( \frac{1 - \hat{x}}{1 - \hat{y}} \right) - \log \left( \frac{\hat{x}}{\hat{y}} \right) \right]
            \\
            \\
            & = & \displaystyle \left( \frac{1 - \hat{x}}{1 - \hat{y}} - \frac{\hat{x}}{\hat{y}} \right) \left[ (1 - 2y) \right.
            \\
            && \qquad \displaystyle \left. - (1 - 2x) \left\{ \frac{\log \left( \displaystyle \frac{1 - \hat{x}}{1 - \hat{y}} \right) - \log \left( \displaystyle \frac{\hat{x}}{\hat{y}} \right)}{\displaystyle \frac{1 - \hat{x}}{1 - \hat{y}} - \frac{\hat{x}}{\hat{y}}} \right\} \right]
        \end{array}
        \label{Eqn: diff_KLD_rho}
    \end{equation}

    From Lemma \ref{Lemma: ROC-Operating-Point-Transformation}, we have
    \begin{equation}
        \displaystyle \frac{\hat{x}}{\hat{y}} \leq \frac{1 - \hat{x}}{1 - \hat{y}}.
    \end{equation}
    In other words, $\displaystyle \frac{1 - \hat{x}}{1 - \hat{y}} - \frac{\hat{x}}{\hat{y}} \geq 0$. Therefore, the sign of $\displaystyle \frac{d D_B}{d\rho}$ does not depend on $\displaystyle \frac{1 - \hat{x}}{1 - \hat{y}} - \frac{\hat{x}}{\hat{y}}$.

    Also, using the properties of the $\log(\cdot)$ function, 
    we have
    \begin{equation}
        \displaystyle \frac{1 - \hat{y}}{1 - \hat{x}} \leq \frac{\log \left( \displaystyle \frac{1 - \hat{x}}{1 - \hat{y}} \right) - \log \left( \displaystyle \frac{\hat{x}}{\hat{y}} \right)}{\displaystyle \frac{1 - \hat{x}}{1 - \hat{y}} - \frac{\hat{x}}{\hat{y}}} \leq \frac{\hat{y}}{\hat{x}}.
        \label{Eqn: logarithm-ineq}
    \end{equation}

    Substituting Equation (\ref{Eqn: logarithm-ineq}) in Equation (\ref{Eqn: diff_KLD_rho}), we have
    \begin{equation}
        \begin{array}{lcl}
            \displaystyle \left( \frac{1 - \hat{x}}{1 - \hat{y}} - \frac{\hat{x}}{\hat{y}} \right)^{-1} \frac{dD_B}{d\rho} & \leq & \displaystyle (1 - 2y)  - (1 - 2x) \left\{ \frac{1 - \hat{y}}{1 - \hat{x}} \right\}
            \\
            \\
            \displaystyle \left( \frac{1 - \hat{x}}{1 - \hat{y}} - \frac{\hat{x}}{\hat{y}} \right)^{-1} \frac{dD_B}{d\rho} & \leq & \displaystyle \frac{-(y - x)}{1 - \hat{x}}
        \end{array}
    \end{equation}

    Since $\displaystyle \frac{dD_B}{d \rho} \leq 0$, $D_B$ is a monotonically decreasing function of $\rho$, for all $\rho \in [0,\frac{1}{2}]$.

\end{proof}

Having analyzed the impact of BSCs on the ROC, let us now shift our focus on finding those quantizers that maximize the KL Divergence at the sensor or the FC. Given any operating point $A = (x,y)$ at the sensor, we investigate the behavior of $D_A$ with respect to $y$, for a fixed value of $x$.

\begin{lemma-app}
    The optimal quantizer always lies on the boundary of the set of all feasible quantizer designs.
    \label{Lemma: Boundary - Optimal Quantizer}
\end{lemma-app}
\begin{proof}
    For a fixed value of $x$, we differentiate $D_A$ with respect to $y$ as follows.
    \begin{equation}
        \displaystyle \left. \frac{d D_A}{dy} \right|_{\mbox{fixed $x$}} =  \frac{1-x}{1-y} - \frac{x}{y}
    \end{equation}

    From Lemma \ref{Lemma: ROC-Operating-Point-Transformation}, we have $\displaystyle \left. \frac{d D_A}{dy} \right|_{\mbox{fixed $x$}} \geq 0$. In other words, $D_A$ is a monotonically increasing function of $y$, for a fixed value of $x$. Hence, we are always interested in quantizer rules whose operating points lie on the boundary of the set of all feasible quantizers.
\end{proof}

In summary, the sensor operating point chosen on the LRT boundary slides towards the point $(\frac{1}{2},\frac{1}{2})$ as the channel deteriorates (increasing $\rho$), which, in turn, degrades the KLD of any decision rule $\gamma$ to zero. Therefore, we address the problem of finding the operating point on the boundary which maximizes $\mathcal{D}_{FC}$, where the boundary is dictated by the Eve's constraint $\mathcal{D}_E = \alpha$ and the boundary of $\boldsymbol\Gamma = \{ \Gamma_1, \cdots, \Gamma_N \}$.


\section{Proof for Theorem \ref{Thrm: conjecture-theorem} \label{Proof: Thrm-conjecture}}
To show that $D$ is a convex function of $x$ in the presence of a constraint on Eve, we investigate the second-order differential of $D$ with respect to $x$.
	
	The closed-form expression for the first-order differential of $D$ with respect to $x$
	\begin{equation}
		\begin{array}{lcl}
			\displaystyle \frac{dD}{dx} & = & \displaystyle \frac{d}{dx} \left[ x \log \frac{x}{y} + (1-x) \log \left( \frac{1-x}{1-y} \right) \right]
			\\
			\\
			& = & \displaystyle \left( \frac{1-x}{1-y} - \frac{x}{y} \right) \frac{dy}{dx} - \left[ \log \left( \frac{1-x}{1-y} \right) - \log \left( \frac{x}{y}\right) \right].
		\end{array}
		\label{Eqn: diff_D_FC_x}
	\end{equation}
	
	The second-order differential of $D$ can therefore be obtained by differentiating Equation (\ref{Eqn: diff_D_FC_x}) with respect to $x$ as follows.
	\begin{equation}
		\begin{array}{lcl}
			\displaystyle \frac{d^2D}{dx^2} & = & \displaystyle \left( \frac{1-x}{1-y} - \frac{x}{y} \right) \frac{d^2y}{dx^2} + \left( \frac{1-x}{(1-y)^2} + \frac{x}{y^2} \right) \left(\frac{dy}{dx}\right)^2
            \\
            \\
            && \quad \displaystyle - 2 \left( \frac{1}{y} + \frac{1}{1-y} \right) \frac{dy}{dx} + \left( \frac{1}{x} + \frac{1}{1-x} \right).
		\end{array}
		\label{Eqn: diff2_D_FC_x}
	\end{equation}
	
	Note that the first term in Equation \eqref{Eqn: diff2_D_FC_x} can be rewritten as follows.
	
	\begin{equation}
		\begin{array}{l}
			\displaystyle \left( \frac{1-x}{1-y} - \frac{x}{y} \right) \frac{d^2y}{dx^2} \ = \ \displaystyle \frac{\displaystyle \left( \frac{1-x}{1-y} - \frac{x}{y} \right)}{\displaystyle \left( \frac{1-\hat{x}}{1-\hat{y}} - \frac{\hat{x}}{\hat{y}} \right)}  \left( \frac{1-\hat{x}}{1-\hat{y}} - \frac{\hat{x}}{\hat{y}} \right) \frac{d^2y}{dx^2}
			\\
			\\
			\qquad \qquad \ = \ \displaystyle \frac{\hat{y}(1-\hat{y})}{y(1-y)} \cdot \frac{1}{(1-2\rho)} \cdot \left( \frac{1-\hat{x}}{1-\hat{y}} - \frac{\hat{x}}{\hat{y}} \right) \frac{d^2y}{dx^2}
		\end{array}
		\label{Eqn: Temp-1}
	\end{equation}
	
	Note that Equation \eqref{Eqn: Temp-1} allows us to use the necessary condition for the operating point $(x,y)$ to lie on the Eve's constraint curve $D_E = \tilde{\alpha}$, as given in Equation \eqref{Eqn: diff2_y_x}. Therefore, we substitute Equation (\ref{Eqn: diff2_y_x}) from the Lemma \ref{Lemma: Eve_constraint} in Equation (\ref{Eqn: Temp-1}), and use this in Equation (\ref{Eqn: diff2_D_FC_x}) to have the following.
	\begin{equation}
		\displaystyle \frac{d^2D}{dx^2} = T_1 \left(\frac{dy}{dx}\right)^2 - 2 T_2 \frac{dy}{dx} + T_3
		\label{Eqn: diff2_D_FC_x_constraint}
	\end{equation}
	where
	
	\begin{subequations}
		\begin{equation}
			T_1 = \left( \frac{1-x}{(1-y)^2} + \frac{x}{y^2} \right) - \frac{\hat{y}(1-\hat{y})}{y(1-y)} \left( \frac{1-\hat{x}}{(1-\hat{y})^2} + \frac{\hat{x}}{\hat{y}^2} \right)
	    	\label{Eqn: T-1}
	    \end{equation}
	    \begin{equation}
	    	T_2 = \left( \frac{1}{y} + \frac{1}{(1-y)} \right) - \frac{\hat{y}(1-\hat{y})}{y(1-y)} \left( \frac{1}{\hat{y}} + \frac{1}{(1-\hat{y})} \right)
	    	\label{Eqn: T-2}
		\end{equation}
		\begin{equation}
			T_3 = \left( \frac{1}{x} + \frac{1}{(1-x)} \right) - \frac{\hat{y}(1-\hat{y})}{y(1-y)} \left( \frac{1}{\hat{x}} + \frac{1}{(1-\hat{x})} \right).
	    	\label{Eqn: T-3}
		\end{equation}
	\end{subequations}
	
	It is easy to show that $T_2 = 0$.
	
	So, let us first consider $T_1$. Expanding Equation (\ref{Eqn: T-1}), we have
	\begin{equation}
		\begin{array}{lcl}
			T_1 & = & \displaystyle \frac{1}{y^2 (1-y)^2 \hat{y} (1 - \hat{y})} \left[ (x\hat{y} - \hat{x}y) - (x\hat{y}^2 - \hat{x}y^2) \right.
            \\
            \\
            && \qquad \qquad \quad \displaystyle \left. + y\hat{y}\left\{(y - \hat{y}) - 2(x - \hat{x}) + 2(x\hat{y} - \hat{x}y)\right\} \right]
			\\
			\\
			& = & \displaystyle \frac{1}{y^2 (1-y)^2 \hat{y} (1 - \hat{y})} \left[ -\rho(y-x) \right.
            \\
            \\
            && \qquad \quad \displaystyle \left. - \left\{\rho^2 x - \rho y^2 + 2 \rho (1-2 \rho) xy - 2 \rho (1-2 \rho) xy^2\right\} \right.
            \\
            \\
            && \qquad \qquad \qquad \displaystyle \left. + y\hat{y} ( \rho - 2\rho x ) \right]
			\\
			\\
			& = & \displaystyle \frac{\rho (1 - \rho) (y - x) (2y - 1)}{y^2 (1-y)^2 \hat{y} (1 - \hat{y})}
		\end{array}
		\label{Eqn: T-1_expanded}
	\end{equation}
	
	Similarly, expanding Equation \ref{Eqn: T-3} for $T_3$, we have
	\begin{equation}
		\begin{array}{lcl}
			T_3 & = & \displaystyle \frac{1}{y(1-y)} \left[ \frac{y(1-y)}{x(1-x)} - \frac{\hat{y}(1-\hat{y})}{\hat{x}(1-\hat{x})} \right]
			\\
			\\
			& = & \displaystyle \frac{\rho (1 - \rho)}{y(1-y)} \cdot \frac{(y-x)(1-x-y)}{x(1-x)\hat{x}(1 - \hat{x})}
		\end{array}
		\label{Eqn: T-3_expanded}
	\end{equation}

	Substituting Equations (\ref{Eqn: T-1_expanded}) and (\ref{Eqn: T-3_expanded}) in Equation (\ref{Eqn: diff2_D_FC_x_constraint}), we simplify Equation \eqref{Eqn: diff2_D_FC_x_constraint} into the following.
	
	\begin{equation}
		\displaystyle \frac{d^2D}{dx^2} = \frac{\rho(1 - \rho)(y-x)}{y(1-y)} \cdot T_4
		\label{Eqn: diff2_D_FC_x_constraint2}
	\end{equation}
	where
	\begin{equation}
		T_4 = \displaystyle \frac{2y-1}{y\hat{y}(1-y) (1-\hat{y})}\left(\frac{dy}{dx}\right)^2 + \frac{1-x-y}{x\hat{x}(1-x)(1-\hat{x})}.
		\label{Eqn: T-4}
	\end{equation}
	
	 Note that, if $T_4 \geq 0$, $D$ is a convex function of $x$ along the Eve's constraint curve $D_E = \tilde{\alpha}$. Since we are only interested in the region where $y \geq x$ and $\rho < \displaystyle \frac{1}{2}$ for all practical purposes, we restrict our analysis of the sign of $T_4$ in this region.
	
	In order to analyze the sign of $T_4$, we divide the achievable region in the receiver-operating characteristics into three regions, as shown in Figure \ref{Fig: ROC regions}.
	\begin{figure}[!t]
		\centering
	    \includegraphics[trim=12cm 0 0 5cm, clip=true, width=3.3in]{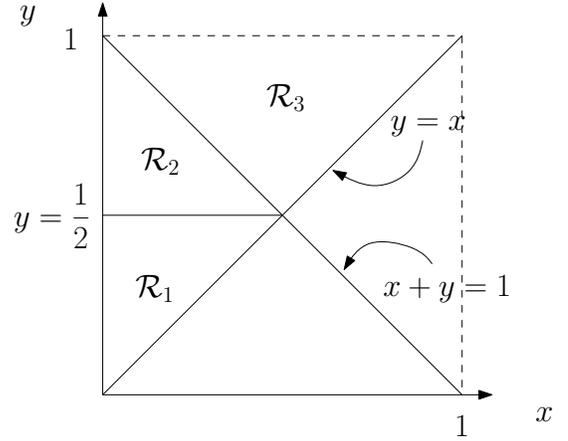}
	    \caption{Partition of ROC into three regions}
	    \label{Fig: ROC regions}
	\end{figure}
	
	\begin{equation}
		\begin{array}{l}
			\mathcal{R}_1: \displaystyle \left( y \leq \frac{1}{2} \right) \& \left( x + y \leq 1 \right)
			\\
			\\
			\mathcal{R}_2: \displaystyle \left( y \geq \frac{1}{2} \right) \& \left( x + y \leq 1 \right)
			\\
			\\
			\mathcal{R}_3: \displaystyle \left( y \geq \frac{1}{2} \right) \& \left( x + y \geq 1 \right).
		\end{array}
	\end{equation}
	
	Obviously, in region $\mathcal{R}_2$, $2y - 1 \geq 0$ and $1 - x - y \geq 0$. Therefore, $\displaystyle \frac{d^2D}{dx^2} \geq 0$. Henceforth, we analyse the sign of $T_4$ in the remaining regions $\mathcal{R}_1$ and $\mathcal{R}_3$.
	
	\paragraph*{Region $\mathcal{R}_1$} In this region, $2y - 1 \leq 0$. Therefore, we use the upper bound on $\frac{dy}{dx}$, presented in Equation (\ref{Eqn: diff_y_x_bound}), to find the sign of $T_4$ as follows.
	
	Substituting Equation (\ref{Eqn: diff_y_x_bound}) in Equation (\ref{Eqn: T-4}), we have
	\begin{equation}
		\begin{array}{l}
			T_4 \ \geq \ \displaystyle \frac{1-x-y}{x\hat{x}(1-x)(1-\hat{x})} - \frac{1-2y}{y\hat{y}(1-y) (1-\hat{y})} \frac{y\hat{y}}{x\hat{x}}
			\\
			\\
			\quad \ = \ \displaystyle \frac{1}{x\hat{x}} \left[ \frac{1-x-y}{(1-x)(1-\hat{x})} - \frac{1-2y}{(1-y)(1-\hat{y})} \right]
			\\
			\\
			\quad \ = \ \displaystyle \frac{1}{x\hat{x}(1-x)(1-\hat{x})(1-y)(1-\hat{y})} \cdot
            \\
            \\
            \quad \left[ (1-x-y)(1-y)(1-\hat{y}) - (1-y)(1-x)(1-\hat{x}) \right.
            \\
            \\
            \qquad \qquad \qquad \qquad \qquad \qquad \left. + y(1-x)(1-\hat{x}) \right]
		\end{array}
		\label{Eqn: T-4_R1}
	\end{equation}
	
	Equation (\ref{Eqn: T-4_R1}) can be rearranged as follows.
	\begin{equation}
		T_4 \geq \displaystyle \frac{(y-x) \left[ y (1-\rho) + (1-2\rho) \left\{ (2y-1)(1-x) - y^2 \right\} \right]}{x\hat{x}(1-x)(1-\hat{x})(1-y)(1-\hat{y})}
		\label{Eqn: T-4_R1_expanded}
	\end{equation}
	
	Since $1-x-y \geq 0$ in region $\mathcal{R}_1$, we have $1-x \geq y$. Therefore, substituting this inequality in Equation (\ref{Eqn: T-4_R1_expanded}), we have
	\begin{equation}
		\begin{array}{lcl}
			T_4 & \geq & \displaystyle \frac{(y-x) \left[ y (1-\rho) + (1-2\rho) \left\{ (2y-1)y - y^2 \right\} \right]}{x\hat{x}(1-x)(1-\hat{x})(1-y)(1-\hat{y})}
			\\
			\\
			& = &  \displaystyle \frac{(y-x) \left[ y (1-\rho) + (1-2\rho) y(y-1) \right]}{x\hat{x}(1-x)(1-\hat{x})(1-y)(1-\hat{y})}
			\\
			\\
			& = & \displaystyle \frac{(y-x) y \left[ (1-\rho) + (1-2\rho)(y-1) \right]}{x\hat{x}(1-x)(1-\hat{x})(1-y)(1-\hat{y})}
			\\
			\\
			& = & \displaystyle \frac{(y-x) y \hat{y}}{x\hat{x}(1-x)(1-\hat{x})(1-y)(1-\hat{y})}
			\\
			\\
			& \geq & 0.
		\end{array}
		\label{Eqn: T-4_R1_final}
	\end{equation}
	
	\paragraph*{Region $\mathcal{R}_3$} In this region, since $2y-1 \geq 0$, we use the lower bound on $\frac{dy}{dx}$, presented in Equation (\ref{Eqn: diff_y_x_bound}), in order to find the sign of $T_4$.
	
	Substituting Equation (\ref{Eqn: diff_y_x_bound}) in Equation (\ref{Eqn: T-4}), we have
	\begin{equation}
		\begin{array}{lcl}
			T_4 & \geq & \displaystyle \frac{2y-1}{y\hat{y}(1-y)(1-\hat{y})} \frac{(1-y)(1-\hat{y})}{(1-x)(1-\hat{x})} - \frac{x+y-1}{x\hat{x}(1-x)(1-\hat{x})}
			\\
			\\
			& = & \displaystyle \frac{1}{(1-x)(1-\hat{x})} \left[ \frac{2y-1}{y\hat{y}} - \frac{x+y-1}{x\hat{x}} \right].
			\\
			\\
			& = & \displaystyle \frac{(y\hat{y} - x\hat{x}) - y(y\hat{y} - x\hat{x}) -xy(\hat{y} - \hat{x})}{x\hat{x}y\hat{y}(1-x)(1-\hat{x})}
			\\
			\\
			& = & \displaystyle \frac{1}{x\hat{x}y\hat{y}(1-x)(1-\hat{x})} \left[ (1-y) \left\{ \rho(y-x) \right. \right.
            \\
            \\
            && \qquad \qquad \displaystyle \left. \left. + (1-2\rho) (y^2 - x^2) \right\} -xy (1-2\rho)(y -x) \right]
			\\
			\\
			& = & \displaystyle \frac{1}{x\hat{x}y\hat{y}(1-x)(1-\hat{x})} \left[ (y-x)\left\{\rho(1-y) \right. \right.
            \\
            \\
            && \qquad \qquad \displaystyle \left. \left.  + (1-2\rho)y(1-y) - (1-2\rho)x(1-2y) \right\} \right]
		\end{array}
		\label{Eqn: T-4_R3}
	\end{equation}
	
	Since we are only interested in the region where $y \geq x$, Equation (\ref{Eqn: T-4_R3}) can be lower-bounded as follows.
		\begin{equation}
		\begin{array}{lcl}
			T_4 & \geq & \displaystyle \frac{1}{x\hat{x}y\hat{y}(1-x)(1-\hat{x})} \left[ (y-x)\left\{\rho(1-y) \right. \right.
            \\
            \\
            && \qquad \qquad \displaystyle \left. \left.  + (1-2\rho)x(1-y) - (1-2\rho)x(1-2y) \right\} \right]
			\\
			\\
			& = & \displaystyle \frac{(y-x)\left[\rho(1-y) + (1-2\rho)xy \right]}{x\hat{x}y\hat{y}(1-x)(1-\hat{x})}
			\\
			\\
			& \geq & 0.
		\end{array}
		\label{Eqn: T-4_R3}
	\end{equation}
	
	Hence, for BSCs with $\rho < \frac{1}{2}$, $D$ is a convex function of $x$ along the constraint $D_E = \alpha$.


\bibliographystyle{IEEEtran}
\bibliography{IEEEabrv,references}

\begin{thebibliography}{10}
\providecommand{\url}[1]{#1}
\csname url@samestyle\endcsname
\providecommand{\newblock}{\relax}
\providecommand{\bibinfo}[2]{#2}
\providecommand{\BIBentrySTDinterwordspacing}{\spaceskip=0pt\relax}
\providecommand{\BIBentryALTinterwordstretchfactor}{4}
\providecommand{\BIBentryALTinterwordspacing}{\spaceskip=\fontdimen2\font plus
\BIBentryALTinterwordstretchfactor\fontdimen3\font minus
  \fontdimen4\font\relax}
\providecommand{\BIBforeignlanguage}[2]{{%
\expandafter\ifx\csname l@#1\endcsname\relax
\typeout{** WARNING: IEEEtran.bst: No hyphenation pattern has been}%
\typeout{** loaded for the language `#1'. Using the pattern for}%
\typeout{** the default language instead.}%
\else
\language=\csname l@#1\endcsname
\fi
#2}}
\providecommand{\BIBdecl}{\relax}
\BIBdecl

\bibitem{Book-Swami}
A.~Swami, Q.~Zhao, Y.-W. Hong, and L.~Tong, \emph{Wireless Sensor Networks:
  Signal Processing and Communications}.\hskip 1em plus 0.5em minus 0.4em\relax
  U.S.A: John Wiley \& Sons Ltd., 2007.

\bibitem{Viswanathan1997}
R.~Viswanathan and P.~K. Varshney, ``Distributed detection with multiple
  sensors i. fundamentals,'' \emph{Proc. {IEEE}}, vol.~85, no.~1, pp. 54--63,
  1997.

\bibitem{Blum1997}
R.~S. Blum, S.~A. Kassam, and H.~V. Poor, ``Distributed detection with multiple
  sensors ii. advanced topics,'' \emph{Proc. {IEEE}}, vol.~85, no.~1, pp.
  64--79, 1997.

\bibitem{Book-Varshney}
P.~K. Varshney, \emph{Distributed Detection and Data Fusion}.\hskip 1em plus
  0.5em minus 0.4em\relax Springer, New York, 1997.

\bibitem{Veeravalli2011}
V.~V. Veeravalli and P.~K. Varshney, ``Distributed inference in wireless sensor
  networks,'' \emph{Philosophical Transactions of the Royal Society of London
  A: Mathematical, Physical and Engineering Sciences}, vol. 370, no. 1958, pp.
  100--117, 2011.

\bibitem{Tsitsiklis1985}
J.~Tsitsiklis and M.~Athans, ``On the complexity of decentralized decision
  making and detection problems,'' \emph{IEEE Transactions on Automatic
  Control}, vol.~30, no.~5, pp. 440--446, May 1985.

\bibitem{Hoballah1986}
I.~Y. Hoballah and P.~K. Varshney, ``Neyman-pearson detection wirh distributed
  sensors,'' in \emph{25th IEEE Conference on Decision and Control}, Dec 1986,
  pp. 237--241.

\bibitem{Chair1986}
Z.~Chair and P.~K. Varshney, ``Optimal data fusion in multiple sensor detection
  systems,'' \emph{IEEE Transactions on Aerospace and Electronic Systems}, vol.
  AES-22, no.~1, pp. 98--101, Jan 1986.

\bibitem{Tsitsiklis1986}
J.~Tsitsiklis, ``On threshold rules in decentralized detection,'' in
  \emph{Decision and Control, 1986 25th IEEE Conference on}, Dec 1986, pp.
  232--236.

\bibitem{Longo1990}
M.~Longo, T.~D. Lookabaugh, and R.~M. Gray, ``Quantization for decentralized
  hypothesis testing under communication constraints,'' \emph{IEEE Transactions
  on Information Theory}, vol.~36, no.~2, pp. 241--255, Mar 1990.

\bibitem{Tsitsiklis1993}
J.~Tsitsiklis, ``Extremal properties of likelihood-ratio quantizers,''
  \emph{Communications, IEEE Transactions on}, vol.~41, no.~4, pp. 550--558,
  Apr 1993.

\bibitem{Hashlamoun1996}
W.~A. Hashlamoun and P.~K. Varshney, ``Near-optimum quantization for signal
  detection,'' \emph{IEEE Transactions on Communications}, vol.~44, no.~3, pp.
  294--297, Mar 1996.

\bibitem{Rago1996}
C.~Rago, P.~Willett, and Y.~Bar-Shalom, ``Censoring sensors: a
  low-communication-rate scheme for distributed detection,'' \emph{IEEE
  Transactions on Aerospace and Electronic Systems}, vol.~32, no.~2, pp.
  554--568, April 1996.

\bibitem{Zhang2002}
Q.~Zhang, P.~K. Varshney, and R.~D. Wesel, ``Optimal bi-level quantization of
  i.i.d. sensor observations for binary hypothesis testing,'' \emph{{IEEE}
  Trans. Inf. Theory}, vol.~48, no.~7, pp. 2105--2111, 2002.

\bibitem{Liu2006}
B.~Liu and B.~Chen, ``Channel-optimized quantizers for decentralized detection
  in sensor networks,'' \emph{IEEE Transactions on Information Theory},
  vol.~52, no.~7, pp. 3349--3358, July 2006.

\bibitem{Wang2013}
Y.~Wang and Y.~Mei, ``Quantization effect on the log-likelihood ratio and its
  application to decentralized sequential detection,'' \emph{Signal Processing,
  IEEE Transactions on}, vol.~61, no.~6, pp. 1536--1543, March 2013.

\bibitem{Aysal2008}
T.~C. Aysal and K.~E. Barner, ``Sensor data cryptography in wireless sensor
  networks,'' \emph{{IEEE} Trans. Inf. Forensics Security}, vol.~3, no.~2, pp.
  273--289, 2008.

\bibitem{Nadendla-Thesis}
V.~S.~S. Nadendla, ``Secure distributed detection in wireless sensor networks
  via encryption of sensor decisions,'' Master's thesis, Louisiana State
  University, 2009.

\bibitem{Jeon2011}
H.~Jeon, S.~W. McLaughlin, and J.~Ha, ``Cooperative secure transmission for
  distributed detection in wireless sensor networks,'' in \emph{Proc. IEEE 54th
  Int Circuits and Systems (MWSCAS) Midwest Symp}, 2011, pp. 1--4.

\bibitem{Marano2009a}
S.~Marano, V.~Matta, and P.~K. Willett, ``Distributed detection with censoring
  sensors under physical layer secrecy,'' \emph{{IEEE} Trans. Signal Process.},
  vol.~57, no.~5, pp. 1976--1986, May 2009.

\bibitem{Nadendla2010a}
V.~S.~S. Nadendla, H.~Chen, and P.~K. Varshney, ``Secure distributed detection
  in the presence of eavesdroppers,'' in \emph{Proc. Conf Signals, Systems and
  Computers (ASILOMAR) Record of the Forty Fourth Asilomar Conf}, 2010, pp.
  1437--1441.

\bibitem{Cover}
T.~M. Cover and J.~A. Thomas, \emph{Elements of Information Theory}.\hskip 1em
  plus 0.5em minus 0.4em\relax U.S.A: John Wiley and Sons Inc., 2006.

\bibitem{Book-Rockafeller1996}
R.~T. Rockafeller, \emph{Convex Analysis}, ser. Princeton Landmarks in
  Mathematics and Physics.\hskip 1em plus 0.5em minus 0.4em\relax Princeton
  University Press, 1996.

\bibitem{Book-Boyd2004}
S.~P. Boyd and L.~Vandenberghe, \emph{Convex Optimization}.\hskip 1em plus
  0.5em minus 0.4em\relax Cambridge University Press, 2004.

\bibitem{Book-Bellman}
R.~Bellman, \emph{Dynamic Programming}.\hskip 1em plus 0.5em minus 0.4em\relax
  Dover Publications, 2003.

\end{thebibliography}


\end{document}